\documentclass{article} 
\usepackage{iclr2027_conference,times}

\usepackage{amsmath,amsfonts,bm}

\def\eqref#1{equation~\ref{#1}}

\def\1{\bm{1}}

\DeclareMathAlphabet{\mathsfit}{\encodingdefault}{\sfdefault}{m}{sl}
\SetMathAlphabet{\mathsfit}{bold}{\encodingdefault}{\sfdefault}{bx}{n}

\usepackage{hyperref}
\usepackage{url}
\usepackage{graphicx}
\usepackage{booktabs}
\usepackage{amsmath}
\usepackage{amssymb}
\usepackage{multirow}
\usepackage{xcolor}

\title{Where Does Detection Fail? A Residual-Evidence Diagnostic Workflow for Industrial Control Systems}

\author{
Peiran Shi, Jian Xiang, Xiang Zhang, and Chenglong Fu \\
\textit{University of North Carolina at Charlotte} \\
\texttt{\{pshi, jian.xiang, xzhang46, chenglong.fu\}@charlotte.edu}
}

\iclrfinalcopy 

\begin{document}

\maketitle

\begin{abstract}
Residual-based anomaly detectors for industrial control systems (ICS)
convert model residuals into anomaly scores and alarms, but final
detection metrics give little guidance on which component to examine
when detection performance is poor. We present a retrospective diagnostic
workflow built around covariance-normalized residual energy (CNRE),
a measurement based on the squared Mahalanobis distance. The workflow
asks whether the normal reference still describes the evaluation
recording, whether attacks separate from normal operation under a stated
reference, and how scoring and alarm rules use fixed residuals.
Across fifteen detector--dataset pairs (five detectors on SWaT, WADI,
and HAI), reference updates improve separation in some pairs but degrade
it in others, and a retrospective policy that updates only where a
label-free screen recommends it raises mean AUC from 0.751 (never
updating) and 0.759 (always updating) to 0.809; on the twelve non-GDN
pairs, which did not inform its boundaries, it gives 0.805 versus 0.759
and 0.747. Rescoring CNRE computed with a fixed reference reproduces most
of the AUC gain of a reference update on WADI but not on SWaT.
A detailed GDN case study connects the diagnostics:
normal-only adaptation reduces normal prediction error by
27.3\%, but the direction of its effect on attack separation
depends on the reference protocol.
With residuals fixed, stronger score ranking does not ensure
that normal-calibrated alarms preserve their nominal
false-positive rate.
These results support choosing intervention stages and interpreting
their outcomes through separate measurements of reference fit, attack
evidence, and alarm calibration.

\end{abstract}

\section{Introduction}
\label{sec:introduction}

Industrial control systems (ICS) rely on anomaly detectors to
identify abnormal behavior in physical processes such as water
treatment, manufacturing, and energy systems
\citep{nist2023otsecurity}.
Many detectors learn expected behavior, compare it with observations,
and convert the resulting deviation into anomaly scores and alarms.
When an attack is missed, the final detection result reveals little
about which stage limited detection. The normal reference may no longer
describe current operation, the deviations may separate attacks weakly,
or downstream scoring and alarm decisions may use the available
evidence poorly.

\begin{figure}[t]
\centering
\includegraphics[width=\linewidth]{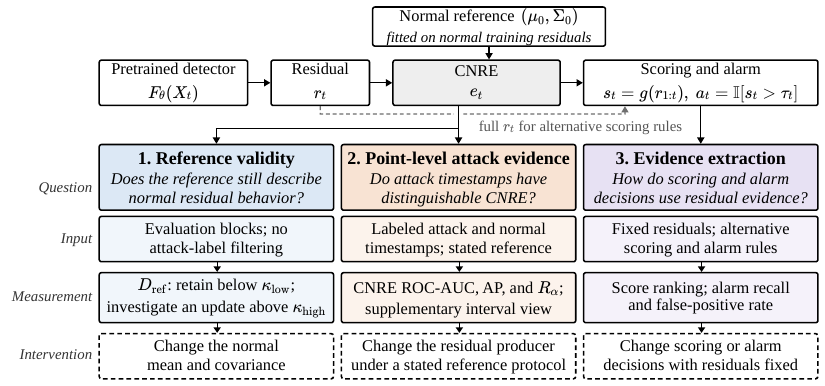}
\caption{Residual-evidence diagnostic workflow. Diagnostics are examined in
the numbered order, and the reference screen is repeated after a model or
reference-protocol change. Dashed boxes are candidate interventions, each
evaluated by its own controlled comparison.}
\label{fig:diagnostic_workflow}
\end{figure}

Component-level and benchmark studies show that model, scoring,
thresholding, and protocol choices can change reported performance
substantially \citep{garg2022evaluation,kim2022rigorous}. We ask a
complementary question: which controlled comparisons distinguish a
mismatched normal reference, a change in attack--normal separability,
and a limitation in downstream scoring or alarm-threshold calibration.
We present a retrospective diagnostic workflow built from such
comparisons.
The workflow identifies comparisons that can explain intervention
outcomes; it does not assign a unique cause to each missed alarm.

The workflow applies to detectors that expose a \emph{residual} before
the alarm, such as prediction errors, reconstruction errors,
innovations, and constraint violations
\citep{deng2021gdn,tuli2022tranad,feng2021nsibf,wolsing2025geco}.
We measure residuals with \emph{covariance-normalized residual energy}
(CNRE), which accounts for the normal scale and covariance of each
detector's residuals. Its quadratic form is the established
Mahalanobis construction \citep{mahalanobis1936,qin2003spm}; we use it
as a common measurement interface rather than as a new statistic.
CNRE supports three diagnostic questions
(Figure~\ref{fig:diagnostic_workflow}): whether the reference still
describes normal residual behavior (\emph{reference validity}), whether
attack timestamps have distinguishable CNRE (\emph{point-level attack
evidence}), and how scoring and alarm decisions use residual evidence
(\emph{evidence extraction}).

We evaluate the workflow across five detectors on SWaT, WADI, and
HAI \citep{mathur2016swat,ahmed2017wadi,shin2021hai}. Updating never
helps a pair that the label-free reference screen retains, and updating
only where the screen recommends it raises mean AUC from $0.751$ (never
updating) and $0.759$ (always updating) to $0.809$. Read after the
reference decision, attack evidence separates pairs that the update
recovers from pairs that stay weak under a suitable reference. Rescoring
reproduces most of the AUC gain of a reference update on WADI but not on
SWaT. A controlled case study on one detector shows that better normal
prediction does not determine attack separation, that the sign of a
model change depends on its reference protocol (CNRE average precision
$0.573$ falls to $0.452$ or rises to $0.593$), and that stronger ranking
does not ensure calibrated alarms.

\paragraph{Contributions.}
\begin{itemize}
    \item \textbf{Workflow.} CNRE connects reference validity, point-level
    attack evidence, and evidence extraction, each with its own
    measurement and controlled comparison.
    \item \textbf{Cross-detector evidence.} On fifteen detector--dataset
    pairs, a label-free screen indicates where a reference update can
    help, a threshold-only control separates alarm failures that a
    threshold fixes from those that also need a reference update, and
    rescoring reproduces most of the update's AUC gain on WADI but not on
    SWaT.
    \item \textbf{Case study.} Reference fitting can reverse the measured
    effect of a model intervention, and stronger ranking does not ensure
    calibrated alarms.
\end{itemize}

\section{Related Work}
\label{sec:related_work}

\subsection{ICS Anomaly Detection}

ICS anomaly detection includes process-aware monitoring and
data-driven modeling. Process-aware methods examine whether
observations agree with expected operating states or process
behavior. State-aware detection incorporates system state into
monitoring \citep{ghaeini2018stateaware}, while GeCo automatically
derives process models from benign operation
\citep{wolsing2025geco}.

Data-driven detectors learn temporal and cross-variable relationships:
GDN predicts with a learned sensor graph \citep{deng2021gdn}, FuSAGNet
combines sparse latent graphs with reconstruction and forecasting
\citep{han2022fusagnet}, TranAD reconstructs with Transformers
\citep{tuli2022tranad}, and NSIBF combines neural system identification
with Bayesian filtering \citep{feng2021nsibf}. Together with GeCo, their
different residuals test whether one measurement can organize comparisons
across detector families.

\subsection{Analysis and Evaluation of Anomaly Detectors}

Component-level studies show that detection performance depends
on more than the residual-producing model.
\citet{garg2022evaluation} separate the model, scoring function,
and thresholding function, and evaluate alternative combinations.
In their experiments, scoring choices strongly affect performance, and
adaptive scoring can lose sensitivity when it incorporates anomalous
behavior.
Their anomaly-diagnosis task identifies affected channels, whereas our
diagnostic target is the detector pipeline. We build on
component-level evaluation by asking which controlled comparisons
distinguish reference mismatch, weak attack--normal separability,
and scoring or calibration limitations, with CNRE as the common
measurement across these questions.

Benchmark studies further show that evaluation choices can obscure
what a detector has learned.
\citet{wu2023benchmarks} identify problems in widely used
time-series anomaly-detection benchmarks.
\citet{kim2022rigorous} show that point adjustment can inflate
reported performance and distort comparisons.
More recently, TAB provides unified benchmarking and examines
inconsistencies in evaluation pipelines \citep{qiu2025tab}.
These studies motivate explicit reporting of data access,
temporal alignment, scoring, and threshold selection.
We do not apply point adjustment, and we distinguish retrospective
ROC operating points from alarms calibrated on normal data.
Supplementary event-level measurements describe variation across
attack intervals; the main diagnostic analysis remains point-level.

\subsection{Residual Statistics and Reference Adaptation}

Covariance-normalized residual statistics have established
foundations in multivariate analysis and process monitoring.
Mahalanobis distance measures deviation relative to a reference
mean and covariance \citep{mahalanobis1936}.
Statistical process monitoring uses related measures to detect
departures from normal operation \citep{qin2003spm}.
Shrinkage covariance estimation provides a practical way to
stabilize these calculations \citep{ledoitwolf2004}.
CNRE uses this established statistical construction as the
measurement core of our workflow.

The usefulness of such a measurement depends on its reference.
Research on concept drift studies how changing data distributions
affect learned models and motivate adaptation
\citep{gama2014drift}.
For anomaly detection, updating the reference also creates a
specific concern: abnormal behavior may enter the data used
to define normality. Streaming threshold methods such as SPOT and DSPOT
face the same trade-off: they fit extreme-value thresholds online and
exclude flagged values from the update \citep{siffer2017spot}.
Our reference experiments examine this tension while keeping
the residual-producing detector fixed.
They connect persistent CNRE elevation to observed update outcomes,
including cases where an update that pauses refitting during elevated
alarm activity still loses attack separation.

\section{Problem Setup}
\label{sec:problem_setup}

\paragraph{Detection pipeline.}
We study pretrained ICS anomaly detectors that expose a residual before
producing an alarm. At time \(t\),
we write
\begin{equation}
    \hat{y}_t=F_\theta(X_t),
    \qquad r_t=\psi(y_t,\hat{y}_t)\in\mathbb{R}^{p}.
    \label{eq:residual}
\end{equation}
Here, \(X_t\) is the detector input, which may include a sequence of
observations. The model \(F_\theta\), with pretrained parameters
\(\theta\), produces the prediction or reconstruction \(\hat{y}_t\) of
the corresponding observation \(y_t\). The function \(\psi\) computes
the detector's residual, and \(p\) is its dimension; for a prediction
error, \(\psi(y_t,\hat{y}_t)=y_t-\hat{y}_t\). We also use the term
residual for innovations and constraint violations.
For diagnostic analysis, we fit a separate normal reference
\((\mu_0,\Sigma_0)\) for each detector, the mean and covariance of its
normal training residuals (Section~\ref{sec:cnre}).
Downstream, an \emph{anomaly scoring rule} and an \emph{alarm decision
rule} convert residuals into alarms:
\begin{equation}
 s_t=g(r_{1:t};\phi),\qquad
 a_t=\mathbb{I}[s_t>\tau_t].
 \label{eq:alarm_rule}
\end{equation}
The scoring rule $g$ maps the residual history $r_{1:t}$ to a scalar
score $s_t$ using fitted parameters $\phi$, such as normalization
statistics. The alarm decision compares that score with a threshold
$\tau_t$, which may be fixed or updated from past data; $a_t=1$
denotes an alarm.
Detection can therefore be limited at three stages: the normal reference
\((\mu_0,\Sigma_0)\), the residual-producing model \(F_\theta\), and the
scoring and alarm rules \((g,\tau_t)\).

\paragraph{Data access and scope.}
References and fixed thresholds are fitted on normal data, and the
reference diagnostic reads the evaluation recording without labels; the
other two diagnostics use benchmark labels. All diagnostics are
\emph{retrospective}, computed after a recording is complete, and
\emph{point-level}, with the timestamp as the unit of evaluation.

\section{Residual-Evidence Diagnostic Workflow}
\label{sec:framework}

The workflow uses CNRE to address the three diagnostic questions of
Figure~\ref{fig:diagnostic_workflow}, each with its own data, measurement,
and controlled comparison. We examine the reference first, because a
mismatched reference distorts both attack--normal separability and
normal-calibrated alarms. We then measure attack evidence under a stated
reference and finally examine scoring and alarms with the residuals fixed.

\subsection{CNRE as a Measurement Interface}
\label{sec:cnre}
Residuals need not be zero during normal operation, and their scales and
correlations differ across variables because of prediction error, sensor
noise, and normal process variation. We therefore estimate the reference
mean $\mu_0\in\mathbb{R}^{p}$ and covariance
$\Sigma_0\in\mathbb{R}^{p\times p}$ from normal training residuals, using
Ledoit--Wolf shrinkage for the covariance \citep{ledoitwolf2004}, and
define
\begin{equation}
 e_t=\frac{1}{2}(r_t-\mu_0)^\top\Sigma_0^{-1}(r_t-\mu_0).
 \label{eq:cnre}
\end{equation}
Centering removes the normal offset of $r_t$, and covariance
normalization accounts for normal residual
scales and correlations between variables. The resulting energy is half
the squared Mahalanobis distance \citep{mahalanobis1936}. A deviation in
a normally stable residual pattern therefore contributes more than a
comparable deviation in a pattern that already varies during normal
operation.

An attack that disrupts the detector's learned relationships can raise
CNRE above normal levels; an attack with little effect on the residuals
may not. Normal operation can also produce high CNRE if the reference
becomes unsuitable. When the reference equals the population normal mean
and covariance, the expected energy is $p/2$. We use this value as a
nominal scale for comparing residual dimensions, not as an alarm
threshold. Each detector keeps its own reference, so raw energies are
not compared across detectors; only the $p/2$-normalized statistic of
Eq.~(\ref{eq:reference_validity}) shares decision boundaries, and
Section~\ref{sec:exp_reference} tests that transfer.

\subsection{Diagnostic Questions and Controlled Comparisons}
\label{sec:diagnostic_workflow}

\paragraph{Reference validity.}
\label{sec:reference_validity}
Training data and later operation can differ as equipment conditions,
control settings, or operating regimes change \citep{gama2014drift}.
A fixed reference may then assign high energy to normal behavior.
To screen for persistent elevation without attack labels, we divide the
evaluation recording into $B$ consecutive blocks
$\mathcal{B}_1,\ldots,\mathcal{B}_B$ and compute
\begin{equation}
 D_{\mathrm{ref}}=
 \operatorname{median}_{b=1,\ldots,B}
 \left[
 \frac{\operatorname{median}_{t\in\mathcal{B}_b} e_t}{p/2}
 \right].
 \label{eq:reference_validity}
\end{equation}
The inner median limits the influence of isolated energy spikes, and the
outer median summarizes the recording's typical block. Division by $p/2$
puts different residual dimensions on a common nominal scale. Because
this compares medians with a nominal mean, a value of one is not a
calibrated validity boundary.

$D_{\mathrm{ref}}$ is a screening statistic for persistent energy
elevation, not a test of reference validity. The decision must account
for two risks: retaining a stale reference and updating a suitable
reference with contaminated data. We therefore use two empirical
boundaries $\kappa_{\mathrm{low}}<\kappa_{\mathrm{high}}$: retain the
reference when $D_{\mathrm{ref}}<\kappa_{\mathrm{low}}$, investigate an
update when $D_{\mathrm{ref}}>\kappa_{\mathrm{high}}$, and make no
recommendation between them, where the screen is inconclusive. The
boundaries are set on designated development cases and held fixed for
every subsequent evaluation; Section~\ref{sec:exp_reference} reports
their values and provenance. They define empirical regions, not
significance thresholds. Because the statistic keeps attacks in the
evaluation blocks, sustained attacks can also raise it. An update
recommendation motivates a controlled reference comparison, not a
conclusion that benign drift caused the elevation. The screen applies to
one model and its reference, so it should be repeated after a model
intervention or a change in the reference-fitting protocol.

\paragraph{Point-level attack evidence.}
\label{sec:attack_evidence}
Even under a suitable reference, attacks may not produce
residuals that CNRE separates from normal operation. Let
$\mathcal{A}$ and $\mathcal{N}$ contain the eligible attack and normal
timestamps. We measure CNRE-based attack--normal separability by
point-level ROC-AUC:
\begin{equation}
 S_{\mathrm{atk}}
 =\Pr(e_i>e_j)+\tfrac{1}{2}\Pr(e_i=e_j),
 \qquad i\sim\mathcal{A},\quad j\sim\mathcal{N}.
 \label{eq:point_attack_evidence}
\end{equation}
Here $i$ and $j$ are drawn uniformly and independently from their
timestamp sets, and ties receive half credit \citep{fawcett2006roc}.
We also examine average precision (AP) and recall $R_{\alpha}$ at a
specified false-positive rate (FPR) $\alpha$, read retrospectively from the
labeled evaluation ROC curve; the experiments use $\alpha=1\%$
(Table~\ref{tab:experimental_measurements} in Appendix~\ref{app:protocol}). These measurements
complement overall AUC with precision--recall behavior and detection at
a restricted false-positive rate.

Weak separation concerns scalar CNRE under the stated reference; it does
not show that the residual vector contains no attack information, and it
does not by itself localize a limitation to the residual producer. A
model intervention tests whether a controlled change to the residual
producer improves separation under a stated reference-fitting protocol. A
supplementary interval view (Appendix~\ref{app:events}) replaces
$\mathcal{A}$ with one contiguous labeled attack interval against a
common normal pool to describe heterogeneity within the pooled result.

\paragraph{Evidence extraction.}
\label{sec:evidence_extraction}
The third diagnostic holds the residuals fixed and examines the scoring
and alarm rules in Eq.~(\ref{eq:alarm_rule}). Setting $s_t=e_t$ gives a baseline scoring rule, which
we compare with other rules on the same residuals and timestamps;
alternative rules may use the full residual vector or a history of CNRE
values, and each comparison states its input and fitted parameters.
Score comparisons assess attack--normal ranking; alarm comparisons
assess recall and false positives under a stated calibration policy.

\paragraph{Controlled comparisons.}
\label{sec:controlled_comparisons}
Each intervention is evaluated by a comparison that states what is held
fixed: reference interventions keep the residuals and change
$(\mu_0,\Sigma_0)$; model interventions change the residual producer and
state whether each model keeps the original reference or refits its own;
scoring interventions keep the residuals and reference and change the
score or alarm rule, stating any extra fitting data. No comparison
assumes that a recommended intervention must succeed.

\section{Experimental Evaluation}
\label{sec:experiments}

We first examine the three diagnostics across detector families and
datasets, closing each with the conclusion the diagnostic supports.
We then use one detector (GDN) to illustrate how the workflow
interprets model and scoring interventions.

\subsection{Experimental Setup and Measurements}
\label{sec:experimental_setup}

The cross-detector study covers GDN, FuSAGNet, TranAD, NSIBF, and
GeCo on SWaT, WADI, and HAI, giving fifteen detector--dataset pairs
\citep{deng2021gdn,han2022fusagnet,tuli2022tranad,feng2021nsibf,wolsing2025geco}.
We use archived residuals with each detector's own preprocessing and
temporal alignment (for TranAD, its v3 residual archive), so this is a
diagnostic study of existing outputs, not a retraining comparison.
Initial CNRE references use normal training residuals, and every
intervention is compared with its baseline on the same eligible
timestamps, those scored after the 13-hour warm-up that adaptive
references need. Three evaluation supports appear: these post-warm-up
timestamps with all attacks retained (reference policies,
Table~\ref{tab:attribution}, and rescoring), the full archived test
stream (Appendix~\ref{app:events}), and the case study's own timestamps
(Section~\ref{sec:exp_gdn_case}), whose separate normal fitting splits
mean that its absolute scores need not match the cross-detector study.
Appendix~\ref{app:protocol} defines the measurements
(Table~\ref{tab:experimental_measurements}) and data access; $R_{1\%}$ is
the retrospective recall at 1\% test FPR.

\subsection{Reference Validity Across Detectors and Datasets}
\label{sec:exp_reference}

\paragraph{Normal residual behavior changes across operating periods.}
On SWaT, the training-normal 99th-percentile CNRE threshold flags
78.3--80.3\% of eligible normal test timestamps for GDN, FuSAGNet,
TranAD, and NSIBF, but 1.3\% for GeCo
(Table~\ref{tab:swat_reference_detail}). The size of this shift depends
on the detector even on one dataset, and the measurement does not
identify its physical cause.

\paragraph{Reference diagnosis and intervention outcomes.}
We compute $D_{\mathrm{ref}}$ using one-hour blocks. The three GDN
cases informed $\kappa_{\mathrm{low}}=3$ and $\kappa_{\mathrm{high}}=10$;
these boundaries remain fixed for the other twelve pairs, and no pair
lies between them; any boundaries with
$1.94<\kappa_{\mathrm{low}}\le\kappa_{\mathrm{high}}<62.5$ give the same
fifteen decisions. The screen varies within a dataset
(Table~\ref{tab:attribution}): GeCo is retained on SWaT, and NSIBF and
GeCo on WADI.

Figure~\ref{fig:reference_diagnostic_outcomes} relates these measurements
to guarded updates using a 12-hour trailing window, a one-hour lag,
and hourly refits. Guarding pauses refits during elevated alarm activity;
a rolling update refits every hour regardless, and either variant can
absorb a long attack into the reference.
The comparison retains all eligible attack timestamps after warm-up.
All eight retain-diagnosed pairs lose AUC when updated; four of the
seven update-diagnosed pairs improve, and the other three, all on SWaT,
change by $-0.037$ to $-0.002$. Applying guarded updates only to
the seven update-diagnosed pairs gives mean AUC $0.809$, compared with
$0.751$ for never updating and $0.759$ for always updating. Importantly, this result is not driven by the three GDN cases used to set the boundaries: Over the
twelve non-GDN pairs, these means are $0.805$, $0.759$, and $0.747$.
This policy comparison is retrospective because the screen summarizes
the whole evaluation recording, although it uses no attack labels.

\begin{figure}[t]
\centering
\includegraphics[width=0.72\linewidth]{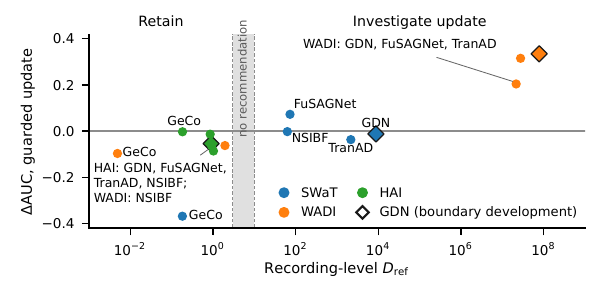}
\caption{Reference diagnosis and guarded-update effects across the fifteen
detector--dataset pairs, on all eligible timestamps after warm-up. Colors
mark datasets and diamonds the GDN cases used to develop the boundaries;
the other points test transfer across detector families on the same
datasets. No pair lies in the shaded 3--10 no-recommendation region.}
\label{fig:reference_diagnostic_outcomes}
\end{figure}

\paragraph{Updating can also remove attack contrast.}
Guarding does not always prevent absorption: it restores FuSAGNet on
SWaT but not TranAD (Table~\ref{tab:swat_reference_detail}). Two scope
checks (Appendix~\ref{app:reference}):
BATADAL is retained for all three GDN seeds and its updates change AUC
by at most $0.055$; on the NoBOOM batch-process pilot
\citep{wagner2025noboom} only 5 of 17 update-diagnosed recordings
improve, so the rule did not transfer reliably in this short-batch pilot.

\paragraph{Diagnostic implication.}
Elevation under the training reference marks the pairs whose reference
should be examined first; updating helped only there, and the guard
matters because a persistent attack can be absorbed.

\subsection{Attack Evidence After the Reference Decision}
\label{sec:exp_attack_evidence}

\begin{table}[t]
\centering
\footnotesize
\setlength{\tabcolsep}{4pt}
\renewcommand{\arraystretch}{0.95}
\caption{Observed response patterns across the fifteen detector--dataset
pairs, on all eligible timestamps after warm-up. AUC (fixed): CNRE ROC-AUC
under the training-normal reference. $\Delta$AUC (update): change from a
guarded update, shown for every pair; the policy applies it only to
update-diagnosed pairs (Table~\ref{tab:reference_diagnosis_main}).
Patterns are defined after the comparisons
(Section~\ref{sec:exp_attack_evidence}) and are not failure causes
identified by the screen.}
\label{tab:attribution}
\begin{tabular}{llrlrrl}
\toprule
Dataset & Detector & $D_{\mathrm{ref}}$ & Screen & AUC (fixed) & $\Delta$AUC (update) & Observed pattern \\
\midrule
SWaT & GDN       & $8.71\times10^{3}$ & Update & 0.814 & $-0.012$ & Alarms miscalibrated \\
     & FuSAGNet  & $72.6$             & Update & 0.848 & $+0.073$ & Improved by update \\
     & TranAD    & $2.15\times10^{3}$ & Update & 0.428 & $-0.037$ & Weak separation \\
     & NSIBF     & $62.5$             & Update & 0.836 & $-0.002$ & Alarms miscalibrated \\
     & GeCo      & $0.181$            & Retain & 0.887 & $-0.369$ & Strong ranking \\
\midrule
WADI & GDN       & $7.77\times10^{7}$ & Update & 0.450 & $+0.335$ & Improved by update \\
     & FuSAGNet  & $2.75\times10^{7}$ & Update & 0.458 & $+0.316$ & Improved by update \\
     & TranAD    & $2.15\times10^{7}$ & Update & 0.583 & $+0.204$ & Improved by update \\
     & NSIBF     & $1.94$             & Retain & 0.886 & $-0.063$ & Strong ranking \\
     & GeCo      & $0.00486$          & Retain & 0.675 & $-0.097$ & Weak separation \\
\midrule
HAI  & GDN       & $0.895$            & Retain & 0.889 & $-0.054$ & Strong ranking \\
     & FuSAGNet  & $1.02$             & Retain & 0.897 & $-0.086$ & Strong ranking \\
     & TranAD    & $0.917$            & Retain & 0.907 & $-0.063$ & Strong ranking \\
     & NSIBF     & $0.850$            & Retain & 0.838 & $-0.013$ & Strong ranking \\
     & GeCo      & $0.182$            & Retain & 0.865 & $-0.003$ & Strong ranking \\
\bottomrule
\end{tabular}
\end{table}

Attack evidence under the fixed reference mixes how well the model
exposes attacks with how stale its reference is, so we read scalar
separation only after the reference comparisons.
Table~\ref{tab:attribution} summarizes the patterns observed after these
comparisons; they describe responses to interventions, not failure causes
identified by the reference screen.
\emph{Improved by update:} the three update-diagnosed WADI detectors gain
$0.204$--$0.335$, and FuSAGNet on SWaT rises from $0.848$ to $0.920$
(Table~\ref{tab:swat_reference_detail}).
\emph{Alarms miscalibrated:} GDN and NSIBF on SWaT already rank attacks at
$0.814$ and $0.836$, and the guarded update leaves ranking nearly
unchanged, yet the training-normal threshold flags about $80\%$ of normal
timestamps. A trailing 99th-percentile threshold on the fixed reference
separates two cases (Table~\ref{tab:swat_reference_detail}). For NSIBF it
lowers the normal FPR to $3.3\%$ with $71.1\%$ recall, better than the
guarded update ($40.2\%$ recall at $4.7\%$), so its alarm failure is
mainly a threshold problem. For GDN it leaves $10.4\%$ FPR at $16.4\%$
recall, whereas the guarded update reaches $1.3\%$ at $35.1\%$, so the
reference update also contributes.
\emph{Weak separation:} TranAD on SWaT stays at or below chance
($0.391$--$0.428$) under all tested references, although without the
longest attack run an update raises it from $0.542$ to $0.666$, so part of
this weakness reflects that attack; GeCo on WADI has AUC $0.675$ and AP
$0.099$ against an attack prevalence of $0.067$ with the retained
reference.
\emph{Strong ranking:} the five HAI pairs, GeCo on SWaT, and NSIBF on
WADI rank between $0.838$ and $0.907$ with a retained reference.

AP for all pairs and $R_{1\%}$ for HAI on the full archived test stream
(Appendix~\ref{app:events}, Tables~\ref{tab:cross_detector_ap}
and~\ref{tab:hai_archived_evidence}) do not always agree with AUC: on HAI,
TranAD has the highest AUC but FuSAGNet the highest AP.

\paragraph{Diagnostic implication.}
These patterns prioritize further comparisons rather than locate a
failure: weak separation calls for other scoring rules or model changes,
strong pooled ranking for low-FPR and calibration checks, and
miscalibrated alarms for separating threshold from reference effects.
None of them excludes limitations at other stages.

\subsection{Evidence Extraction: Scoring and Calibration on Fixed Residuals}
\label{sec:exp_extraction}

\paragraph{Cross-detector scoring intervention.}
To change scoring without touching the reference, trailing-percentile
rescoring keeps the detector, residuals, and CNRE reference fixed and
replaces each energy by its percentile among the energies of a preceding
12-hour window that ends one hour earlier and advances hourly, so it uses
only past observations (Appendix~\ref{app:quantile_control};
Table~\ref{tab:cross_detector_recalibration}).

\begin{table}[t]
\centering
\small
\caption{AUC change from trailing-percentile rescoring relative to fixed
CNRE on matched timestamps, with residuals and reference unchanged. All
keeps every attack; Excl.\ omits each recording's longest attack run.}
\label{tab:cross_detector_recalibration}
\begin{tabular}{lrrrrrr}
\toprule
 & \multicolumn{2}{c}{SWaT} & \multicolumn{2}{c}{WADI} & \multicolumn{2}{c}{HAI} \\
\cmidrule(lr){2-3}\cmidrule(lr){4-5}\cmidrule(lr){6-7}
Detector & All & Excl. & All & Excl. & All & Excl. \\
\midrule
GDN      & $-0.222$ & $+0.028$ & $+0.269$ & $+0.162$ & $+0.004$ & $+0.005$ \\
FuSAGNet & $-0.221$ & $+0.004$ & $+0.276$ & $+0.171$ & $-0.007$ & $-0.007$ \\
TranAD   & $-0.191$ & $+0.004$ & $+0.187$ & $+0.159$ & $+0.008$ & $+0.009$ \\
NSIBF    & $+0.039$ & $+0.028$ & $-0.096$ & $-0.162$ & $+0.003$ & $-0.001$ \\
GeCo     & $-0.061$ & $-0.039$ & $-0.011$ & $-0.017$ & $-0.001$ & $-0.001$ \\
\bottomrule
\end{tabular}
\end{table}

On all timestamps, rescoring reproduces $80$--$91\%$ of the
guarded-update gain for the three update-diagnosed WADI detectors but
lowers AUC by $0.191$--$0.222$ for GDN, FuSAGNet, and TranAD on SWaT,
where the guarded update changes it by $-0.037$ to $+0.073$. Like an
adaptive reference, the trailing window absorbs the ten-hour SWaT attack,
the loss of sensitivity that \citet{garg2022evaluation} describe for
adaptive scoring. Without the longest run, rescoring reproduces
$57$--$86\%$ of the rolling-update gain on WADI but only $2$--$13\%$ on
SWaT. Rescoring thus reproduces part of the point-level AUC improvement of
a reference update on WADI, but not on SWaT under the evaluated protocol;
because only one rescoring rule is tested, the remaining SWaT gap does not
show that covariance changes are necessary. Where no update is
recommended, HAI changes stay within $0.009$, but the other
retain-diagnosed pairs lose up to $0.096$ ($0.162$ without the longest
run).

\paragraph{Diagnostic implication.}
Rescoring after a reference flag thus shows whether the scoring stage
can compensate for drift; its losses on the long SWaT attack motivate
contamination-aware adaptive scorers, and whether the reference guard
transfers remains to be tested. Alarm-threshold calibration, the
remaining extraction question, is examined on HAI in the case study.

\subsection{GDN Case Study: Interpreting Model and Scoring Interventions}
\label{sec:exp_gdn_case}

The following GDN comparisons illustrate how the workflow interprets
model and scoring interventions; they are separate comparisons, not a
policy of successive repairs, and use the case study's HAI evaluation
(Appendix~\ref{app:protocol}).

\begin{figure}[t]
\centering
\includegraphics[width=0.72\linewidth]{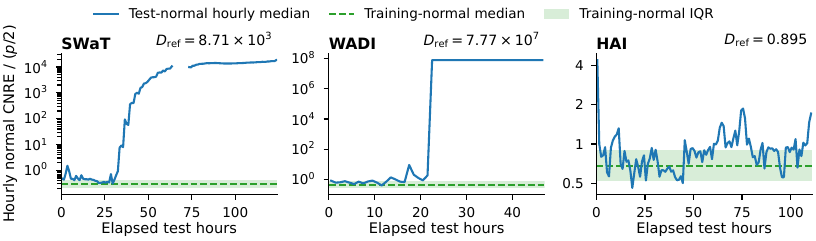}
\caption{Hourly median CNRE of GDN's normal test timestamps (blue) under
the fixed training-normal reference, against the training-normal median
and interquartile range (green), all divided by $p/2$. $D_{\mathrm{ref}}$
uses all evaluation blocks, including attacks; gaps are hours without
normal timestamps.}
\label{fig:normal_cnre_shift}
\end{figure}

\paragraph{Reference.}
Figure~\ref{fig:normal_cnre_shift} shows why the screen flags SWaT and
WADI but not HAI for GDN: SWaT's normal energy rises during the recording
to about $10^{4}$ times the nominal scale, WADI's rises as a step, and
HAI's stays near the training band. On SWaT, rolling updates absorb part
of the ten-hour attack (AUC $0.814$ fixed, $0.771$ rolling), and pausing
refits during elevated alarm activity recovers most of that loss
($0.802$).

\begin{table}[t]
\centering
\small
\caption{HAI adapter under two reference protocols. Adapter values are
five-seed means on the same timestamps as the frozen GDN;
Table~\ref{tab:hai_adapter_reference} adds two test-label references.}
\label{tab:hai_adapter_protocols_main}
\begin{tabular}{lrrrrrr}
\toprule
 & \multicolumn{2}{c}{ROC-AUC} & \multicolumn{2}{c}{AP} & \multicolumn{2}{c}{$R_{1\%}$} \\
\cmidrule(lr){2-3}\cmidrule(lr){4-5}\cmidrule(lr){6-7}
Reference & Frozen & Adapter & Frozen & Adapter & Frozen & Adapter \\
\midrule
Each model's own refit & 0.890 & 0.888 & 0.573 & 0.452 & 0.671 & 0.504 \\
Frozen model's reference & 0.890 & 0.900 & 0.573 & 0.593 & 0.671 & 0.721 \\
\bottomrule
\end{tabular}
\end{table}

\paragraph{Attack evidence.}
On HAI, whose reference the screen does not flag, a residual adapter with
a 16-unit bottleneck, inserted after GDN's node embeddings, is trained to
reduce normal prediction MSE with the base checkpoint frozen, over five
seeds (Appendix~\ref{app:adapter}). Normal MSE falls by 27.3\%, but the
sign of the effect on attack evidence depends on the reference protocol
(Table~\ref{tab:hai_adapter_protocols_main}): AP falls from $0.573$ to
$0.452$ and $R_{1\%}$ from $0.671$ to $0.504$ when each model's reference
is refitted to its own normal residuals, but they rise to $0.593$ and
$0.721$ when both models keep the frozen model's reference. Under the
refitted references, normal energy inflates sharply in two of the five
HAI test recordings and only mildly in the others, whereas the frozen
model's reference raises it evenly (Appendix~\ref{app:adapter}); this is
consistent with refitting amplifying normal-to-test mismatch, but the
comparison does not isolate recentering, covariance reweighting, and
temporal drift. Better normal fit therefore does not
determine attack evidence, a model intervention cannot be read without
its reference protocol, and the reference screen must be repeated after a
model change.
\paragraph{Extraction.}
On the same 40,171 HAI timestamps of the frozen GDN, CNRE raises AUC
over the GDN-style scoring rule from $0.842$ to $0.890$, AP from
$0.377$ to $0.573$, and $R_{1\%}$ from $0.433$ to $0.671$. Yet
thresholds fixed at each score's 99.5th percentile on normal development
data, a nominal FPR of $0.5\%$, give $72.1\%$ recall at $2.19\%$ FPR for
CNRE and $38.4\%$ recall at $0.51\%$ FPR for the GDN-style rule:
stronger ranking does not ensure that a calibrated threshold keeps its
nominal FPR. Label-assisted selection on HAI recordings 1--2 chooses a
robust-sum score, which re-centers and rescales each whitened coordinate
before summing squares; it does not improve held-out AUC on recordings
3--5 ($0.954$ versus $0.957$ for CNRE; Appendix~\ref{app:extraction}).

\paragraph{Outcome for GDN.}
On SWaT, a trailing threshold reduces but does not remove GDN's alarm
miscalibration, and a reference update brings the FPR near its nominal
level while partly absorbing the persistent attack; on WADI the update
restores ranking; on HAI pooled ranking is strong, a model change must
be read with its reference protocol, and alarm-threshold calibration
remains open.

\section{Discussion and Conclusion}
\label{sec:conclusion}

We presented a workflow that uses CNRE to organize controlled comparisons
at three stages of residual-based ICS detection. A label-free screen
identified the pairs whose reference should be examined first, and a
selective update policy outperformed always and never updating; a
threshold-only control separated alarm failures that a threshold fixes
from those that also need a reference update; and rescoring reproduced
most of the update's AUC gain on WADI but not on SWaT. In the GDN case
study, better normal prediction did not determine attack separation, the
sign of a model change depended on its reference protocol, and stronger
ranking did not ensure calibrated false-positive rates.

The reference policy is retrospective; its boundaries were developed on
one detector family and tested on others over the same three datasets,
and they did not transfer reliably in a short-batch pilot. The rescoring
and rolling comparisons are unguarded, and the case study uses one
detector with one pretrained checkpoint per dataset. Natural extensions
are a prospective screen computed on a trailing window and a guard
designed for adaptive scores as well as adaptive references.

\subsection*{AI use statement}
We used generative AI tools to assist the experiments, to support
literature discovery, and to aid writing. For the experiments, they
helped design control analyses, wrote and debugged the analysis and
plotting code, ran it on the saved detector residuals, and summarized
and helped interpret the outputs. They also suggested related work,
helped edit the manuscript and review drafts for clarity and
consistency, and drafted the workflow diagram in
Figure~\ref{fig:diagnostic_workflow}. We have not used generative AI
tools to generate synthetic datasets or to prove mathematical claims.
The authors checked the analysis code, verified every reported number
against the saved outputs, and take responsibility for the full content
of this paper, including text, claims, and artifacts produced with AI
assistance.

\subsection*{Ethics statement}
This work analyzes public ICS security benchmarks (SWaT, WADI, HAI,
BATADAL, and NoBOOM) under their data-access terms. It studies the
detection of attacks and does not introduce new attack techniques.
A diagnostic recommendation can be wrong in either direction, delaying a
needed reference update or triggering an unnecessary one, so
recommendations should be checked with controlled comparisons before
they inform an operational change.

\subsection*{Reproducibility statement}
Appendix~\ref{app:protocol} specifies data access, normal development
splits, eligible targets, label aggregation, and metric definitions.
Appendix~\ref{app:reference} gives the reference-update and guard
settings, Appendix~\ref{app:extraction} the scoring rules and calibration
policies, and Appendix~\ref{app:adapter} the adapter architecture,
training, and reference protocols. The code at
\url{https://anonymous.4open.science/r/iclr2027-residual-evidence-diagnostics-E5D2}
regenerates every table and figure from the saved residuals and checks
each reported value against the regenerated results. The residual
archives include per-timestamp data derived from SWaT and WADI, whose
licenses do not allow redistribution; they are available from the
authors to license holders.

\bibliography{iclr2027_conference}

@techreport{nist2023otsecurity,
  author      = {Stouffer, Keith and Pease, Michael and Tang, CheeYee and Zimmerman, Timothy and Pillitteri, Victoria and Lightman, Suzanne and Hahn, Adam and Saravia, Stephanie and Sherule, Aslam and Thompson, Michael},
  title       = {Guide to Operational Technology ({OT}) Security},
  institution = {National Institute of Standards and Technology},
  number      = {NIST Special Publication 800-82 Rev. 3},
  year        = {2023},
  doi         = {10.6028/NIST.SP.800-82r3}
}

@inproceedings{deng2021gdn,
  author    = {Deng, Ailin and Hooi, Bryan},
  title     = {Graph Neural Network-Based Anomaly Detection in Multivariate Time Series},
  booktitle = {Proceedings of the AAAI Conference on Artificial Intelligence},
  volume    = {35},
  number    = {5},
  pages     = {4027--4035},
  year      = {2021},
  doi       = {10.1609/aaai.v35i5.16523}
}

@inproceedings{han2022fusagnet,
  author    = {Han, Siho and Woo, Simon S.},
  title     = {Learning Sparse Latent Graph Representations for Anomaly Detection in Multivariate Time Series},
  booktitle = {Proceedings of the 28th ACM SIGKDD Conference on Knowledge Discovery and Data Mining},
  pages     = {2977--2986},
  publisher = {ACM},
  year      = {2022},
  doi       = {10.1145/3534678.3539117}
}

@article{tuli2022tranad,
  author  = {Tuli, Shreshth and Casale, Giuliano and Jennings, Nicholas R.},
  title   = {{TranAD}: Deep Transformer Networks for Anomaly Detection in Multivariate Time Series Data},
  journal = {Proceedings of the VLDB Endowment},
  volume  = {15},
  number  = {6},
  pages   = {1201--1214},
  year    = {2022},
  doi     = {10.14778/3514061.3514067}
}

@inproceedings{feng2021nsibf,
  author    = {Feng, Cheng and Tian, Pengwei},
  title     = {Time Series Anomaly Detection for Cyber-Physical Systems via Neural System Identification and {Bayesian} Filtering},
  booktitle = {Proceedings of the 27th ACM SIGKDD Conference on Knowledge Discovery and Data Mining},
  pages     = {2858--2867},
  publisher = {ACM},
  year      = {2021},
  doi       = {10.1145/3447548.3467137}
}

@inproceedings{siffer2017spot,
  author    = {Siffer, Alban and Fouque, Pierre-Alain and Termier, Alexandre and Largou{\"e}t, Christine},
  title     = {Anomaly Detection in Streams with Extreme Value Theory},
  booktitle = {Proceedings of the 23rd ACM SIGKDD International Conference on Knowledge Discovery and Data Mining},
  pages     = {1067--1075},
  publisher = {ACM},
  year      = {2017},
  doi       = {10.1145/3097983.3098144}
}

@inproceedings{wolsing2025geco,
  author    = {Wolsing, Konrad and Wagner, Eric and Lux, Luisa and Wehrle, Klaus and Henze, Martin},
  title     = {{GeCos} Replacing Experts: Generalizable and Comprehensible Industrial Intrusion Detection},
  booktitle = {34th USENIX Security Symposium (USENIX Security 25)},
  pages     = {7153--7172},
  publisher = {USENIX Association},
  year      = {2025}
}

@inproceedings{ghaeini2018stateaware,
  author    = {Ghaeini, Hamid Reza and Antonioli, Daniele and Brasser, Ferdinand and Sadeghi, Ahmad-Reza and Tippenhauer, Nils Ole},
  title     = {State-Aware Anomaly Detection for Industrial Control Systems},
  booktitle = {Proceedings of the 33rd Annual ACM Symposium on Applied Computing},
  pages     = {1620--1628},
  publisher = {ACM},
  year      = {2018},
  doi       = {10.1145/3167132.3167305}
}

@article{garg2022evaluation,
  author  = {Garg, Astha and Zhang, Wenyu and Samaran, Jules and Savitha, Ramasamy and Foo, Chuan-Sheng},
  title   = {An Evaluation of Anomaly Detection and Diagnosis in Multivariate Time Series},
  journal = {IEEE Transactions on Neural Networks and Learning Systems},
  volume  = {33},
  number  = {6},
  pages   = {2508--2517},
  year    = {2022},
  doi     = {10.1109/TNNLS.2021.3105827}
}

@article{mahalanobis1936,
  author  = {Mahalanobis, Prasanta Chandra},
  title   = {On the Generalised Distance in Statistics},
  journal = {Proceedings of the National Institute of Sciences of India},
  volume  = {2},
  number  = {1},
  pages   = {49--55},
  year    = {1936}
}

@inproceedings{mathur2016swat,
  author    = {Mathur, Aditya P. and Tippenhauer, Nils Ole},
  title     = {{SWaT}: A Water Treatment Testbed for Research and Training on {ICS} Security},
  booktitle = {2016 International Workshop on Cyber-Physical Systems for Smart Water Networks (CySWater)},
  pages     = {31--36},
  publisher = {IEEE},
  year      = {2016},
  doi       = {10.1109/CySWater.2016.7469060}
}

@inproceedings{ahmed2017wadi,
  author    = {Ahmed, Chuadhry Mujeeb and Palleti, Venkata Reddy and Mathur, Aditya P.},
  title     = {{WADI}: A Water Distribution Testbed for Research in the Design of Secure Cyber Physical Systems},
  booktitle = {Proceedings of the 3rd International Workshop on Cyber-Physical Systems for Smart Water Networks},
  pages     = {25--28},
  publisher = {ACM},
  year      = {2017},
  doi       = {10.1145/3055366.3055375}
}

@inproceedings{shin2021hai,
  author    = {Shin, Hyeok-Ki and Lee, Woomyo and Yun, Jeong-Han and Min, Byung-Gil},
  title     = {Two {ICS} Security Datasets and Anomaly Detection Contest on the {HIL}-Based Augmented {ICS} Testbed},
  booktitle = {Proceedings of the 14th Cyber Security Experimentation and Test Workshop},
  pages     = {36--40},
  publisher = {ACM},
  year      = {2021},
  doi       = {10.1145/3474718.3474719}
}

@article{taormina2018batadal,
  author  = {Taormina, Riccardo and Galelli, Stefano and Tippenhauer, Nils Ole and Salomons, Elad and Ostfeld, Avi and Eliades, Demetrios G. and others},
  title   = {Battle of the Attack Detection Algorithms: Disclosing Cyber Attacks on Water Distribution Networks},
  journal = {Journal of Water Resources Planning and Management},
  volume  = {144},
  number  = {8},
  pages   = {04018048},
  year    = {2018},
  doi     = {10.1061/(ASCE)WR.1943-5452.0000969}
}

@article{ledoitwolf2004,
  author  = {Ledoit, Olivier and Wolf, Michael},
  title   = {A Well-Conditioned Estimator for Large-Dimensional Covariance Matrices},
  journal = {Journal of Multivariate Analysis},
  volume  = {88},
  number  = {2},
  pages   = {365--411},
  year    = {2004},
  doi     = {10.1016/S0047-259X(03)00096-4}
}

@article{gama2014drift,
  author  = {Gama, Jo{\~a}o and {\v Z}liobait{\.e}, Indr{\.e} and Bifet, Albert and Pechenizkiy, Mykola and Bouchachia, Abdelhamid},
  title   = {A Survey on Concept Drift Adaptation},
  journal = {ACM Computing Surveys},
  volume  = {46},
  number  = {4},
  pages   = {44:1--44:37},
  year    = {2014},
  doi     = {10.1145/2523813}
}

@article{fawcett2006roc,
  author  = {Fawcett, Tom},
  title   = {An Introduction to {ROC} Analysis},
  journal = {Pattern Recognition Letters},
  volume  = {27},
  number  = {8},
  pages   = {861--874},
  year    = {2006},
  doi     = {10.1016/j.patrec.2005.10.010}
}

@inproceedings{kingma2015adam,
  author    = {Kingma, Diederik P. and Ba, Jimmy},
  title     = {Adam: A Method for Stochastic Optimization},
  booktitle = {International Conference on Learning Representations},
  year      = {2015},
  url       = {https://arxiv.org/abs/1412.6980}
}

@inproceedings{wagner2025noboom,
  author    = {Wagner, Dennis and Hartung, Fabian and Arweiler, Justus and Muraleedharan, Aparna and Jungjohann, Indra and Nair, Arjun and Reithermann, Steffen and Schulz, Ralf and Bortz, Michael and Neider, Daniel and Leitte, Heike and Pfeffinger, Joachim and Mandt, Stephan and Fellenz, Sophie and Katz, Torsten and Jirasek, Fabian and Burger, Jakob and Hasse, Hans and Kloft, Marius},
  title     = {{NoBOOM}: Chemical Process Datasets for Industrial Anomaly Detection},
  booktitle = {Advances in Neural Information Processing Systems},
  volume    = {38},
  year      = {2025},
  doi       = {10.52202/085713-3499},
  url       = {https://proceedings.nips.cc/paper_files/paper/2025/hash/96f8c5e879c339dae55e6c2188b02a33-Abstract-Datasets_and_Benchmarks_Track.html}
}

@inproceedings{kim2022rigorous,
  author    = {Kim, Siwon and Choi, Kukjin and Choi, Hyun-Soo
               and Lee, Byunghan and Yoon, Sungroh},
  title     = {Towards a Rigorous Evaluation of Time-Series
               Anomaly Detection},
  booktitle = {Proceedings of the AAAI Conference on Artificial
               Intelligence},
  volume    = {36},
  number    = {7},
  pages     = {7194--7201},
  year      = {2022},
  doi       = {10.1609/aaai.v36i7.20680}
}

@article{wu2023benchmarks,
  author  = {Wu, Renjie and Keogh, Eamonn J.},
  title   = {Current Time Series Anomaly Detection Benchmarks
             Are Flawed and Are Creating the Illusion of Progress},
  journal = {IEEE Transactions on Knowledge and Data Engineering},
  volume  = {35},
  number  = {3},
  pages   = {2421--2429},
  year    = {2023},
  doi     = {10.1109/TKDE.2021.3112126}
}

@article{qiu2025tab,
  author  = {Qiu, Xiangfei and Li, Zhe and Qiu, Wanghui
             and Hu, Shiyan and Zhou, Lekui and Wu, Xingjian
             and Li, Zhengyu and Guo, Chenjuan and Zhou, Aoying
             and Sheng, Zhenli and Hu, Jilin
             and Jensen, Christian S. and Yang, Bin},
  title   = {{TAB}: Unified Benchmarking of Time Series
             Anomaly Detection Methods},
  journal = {Proceedings of the VLDB Endowment},
  volume  = {18},
  number  = {9},
  pages   = {2775--2789},
  year    = {2025},
  doi     = {10.14778/3746405.3746407}
}

@article{qin2003spm,
  author  = {Qin, S. Joe},
  title   = {Statistical Process Monitoring: Basics and Beyond},
  journal = {Journal of Chemometrics},
  volume  = {17},
  number  = {8--9},
  pages   = {480--502},
  year    = {2003},
  doi     = {10.1002/cem.800}
}
\bibliographystyle{iclr2027_conference}

\clearpage
\appendix
\section{Data Access and Evaluation Protocol}
\label{app:protocol}

The experiments follow the three diagnostic questions in the main text.
Reference interventions change the normal residual reference while keeping
the detector fixed. The adapter case study changes a restricted part of the
residual-producing model. Scoring comparisons keep the residuals fixed and
change their conversion into scores or alarms. This separation makes the
location of each intervention explicit.

\paragraph{Residual sources and benchmark support.}
The cross-detector study uses archived residuals from GDN, FuSAGNet, TranAD,
NSIBF, and GeCo on SWaT, WADI, and HAI
\citep{deng2021gdn,han2022fusagnet,tuli2022tranad,feng2021nsibf,wolsing2025geco}.
The datasets are described by \citet{mathur2016swat},
\citet{ahmed2017wadi}, and \citet{shin2021hai}, respectively.
We retain each detector's archived preprocessing, residual dimensions,
sampling interval, and timestamp alignment. TranAD results consistently
use the v3 residual archive. These experiments compare reference choices
within each detector; they do not retrain every detector under a common
architecture or preprocessing pipeline.

The adapter and detailed extraction analyses use the GDN case
study, which fixes one pretrained checkpoint and its evaluation timestamps. Its normal fitting split differs from the full training-normal
reference used in the archived reference study. Consequently, the two GDN
evaluations need not give identical scores or AUCs. Every intervention
comparison uses the same eligible timestamps for its paired models or
reference choices. We do not interpret differences between unmatched
evaluation supports as intervention effects.

\paragraph{GDN preprocessing and eligible targets.}
In the GDN case study, inputs aggregate ten consecutive raw samples. Each
processed target is labeled as an attack if any constituent raw sample
is labeled as an attack. The model uses five processed observations to
predict the next observation. HAI has 79 variables and 40,171 eligible
test targets, including 935 attack targets and 39,236 normal targets.
SWaT has 51 variables and 44,986 test targets, including 5,493 attack
targets and 39,493 normal targets.

Windows crossing gaps within recordings or boundaries between recordings
are removed from the HAI evaluation. A normal target whose input history
contains an attack remains a normal target in the primary point-level
metrics. The clean-normal pool used to fit the held-out scoring rules
additionally excludes such histories. We distinguish these denominators
when reporting false-positive rates in Appendix~\ref{app:extraction}.

\paragraph{Normal development splits.}
The normal recordings are divided chronologically into fitting,
early-stopping, configuration-selection, and retrospective-audit regions.
Five candidate window starts are purged between adjacent regions.
Table~\ref{tab:normal_development_splits} gives the retained counts.
The fitting region supplies gradient updates and the normal residual
reference. The early-stopping region selects the epoch and calibrates
the fixed alarms described in Appendix~\ref{app:extraction}.
Normal prediction MSE in the main adapter comparison is measured on
the configuration-selection split.

\begin{table}[ht]
\centering
\small
\caption{Retained normal development windows for the GDN case study.
The HAI fitting count excludes 12 windows that cross recording gaps.}
\label{tab:normal_development_splits}
\begin{tabular}{lrrrr}
\toprule
Dataset & Fitting & Early stopping & Config.\ selection & Audit \\
\midrule
HAI  & 55,281 & 9,211 & 9,205 & 18,431 \\
SWaT & 28,509 & 4,746 & 4,747 & 9,498 \\
\bottomrule
\end{tabular}
\end{table}

The inherited input scaler was fitted on the full raw normal recording,
and these normal data had been inspected during earlier development.
The precise original checkpoint validation split is not fully resolved
by the archived metadata. These regions therefore support a normal-only
development comparison, but they are not pristine prospective holdouts.
The adapter objective does not use attack labels.

\paragraph{Where labels enter the analysis.}
The initial CNRE reference uses normal training residuals. The
recording-level reference statistic uses evaluation blocks without
removing attacks. Labels are used retrospectively to isolate normal
targets in the drift visualization and to measure attack--normal
separation. The adapter training uses normal data only. In contrast,
the held-out scoring experiment uses labels from HAI recordings 1--2,
the scorer-selection recordings, to select a scoring rule and identify
its clean fitting samples. Recordings 3--5 evaluate the frozen choice.
The test-regime and oracle
references in Appendix~\ref{app:adapter} also use test labels and are
retrospective diagnostic comparisons.

The trailing-percentile control in
Section~\ref{sec:exp_extraction} uses past evaluation energies,
including attacks, without label filtering.
Labels are used only to evaluate the resulting scores and, for the
Excl.\ columns of Table~\ref{tab:cross_detector_recalibration}, to
exclude the longest attack run from metric calculation. This exclusion
does not remove observations from
the scoring history.

\paragraph{Measurements.}
Table~\ref{tab:experimental_measurements} summarizes how each measurement
connects the experiments to the diagnostics; the definitions below make
them precise.

\begin{table}[ht]
\centering
\footnotesize
\setlength{\tabcolsep}{4pt}
\renewcommand{\arraystretch}{0.95}
\caption{Measurements used to connect the experiments to the diagnostics.
Higher values indicate stronger ranking or detection for AUC, AP,
and recall; neither CNRE magnitude nor $D_{\mathrm{ref}}$ alone
establishes an attack.}
\label{tab:experimental_measurements}
\begin{tabular}{p{0.22\linewidth}p{0.70\linewidth}}
\toprule
Measurement & Meaning and role \\
\midrule
CNRE, $e_t$ & Residual deviation relative to normal scale and covariance. \\
$D_{\mathrm{ref}}$ & Persistent CNRE elevation across evaluation blocks; motivates checking the reference. \\
CNRE ROC-AUC & How often attack CNRE ranks above normal CNRE. \\
Average precision (AP) & Precision summarized over increasing recall; complements ROC-AUC. \\
$R_{1\%}$ & Retrospective recall at 1\% test FPR, interpolated from the test ROC curve. \\
Fixed-threshold recall / FPR & Attack fraction detected / normal fraction flagged after normal-data calibration. \\
Normal MSE & Prediction fit on the normal selection split; not an attack-detection metric. \\
\bottomrule
\end{tabular}
\end{table}

\paragraph{Point-level metrics.}
Let $\mathcal{T}$ be the eligible evaluation timestamps, $\ell_t\in\{0,1\}$
the label at timestamp $t$, and $s_t$ a score whose larger values indicate
stronger anomaly evidence. An alarm at threshold $\tau$ is
$a_t(\tau)=\mathbb{I}[s_t>\tau]$. Its recall and false-positive rate are
\begin{equation}
 R(\tau)=\frac{\sum_{t\in\mathcal{T}}\ell_t a_t(\tau)}
                   {\sum_{t\in\mathcal{T}}\ell_t},
 \qquad
 F(\tau)=\frac{\sum_{t\in\mathcal{T}}(1-\ell_t)a_t(\tau)}
                   {\sum_{t\in\mathcal{T}}(1-\ell_t)}.
 \label{eq:app_point_rates}
\end{equation}
Thus, $R$ measures the fraction of attack targets detected, while $F$
measures the fraction of normal targets flagged. We do not apply point
adjustment or expand an alarm to cover its entire attack interval.

ROC-AUC summarizes the ranking obtained as $\tau$ varies
\citep{fawcett2006roc}. Average precision (AP) sums precision at each
distinct score threshold, weighted by the corresponding increase in
recall. It therefore summarizes the precision--recall curve rather than
the ROC curve. Tied scores are handled together. The reported
$R_{1\%}$ is recall at $F=0.01$, obtained by linear interpolation of
the empirical test ROC curve. It uses test labels retrospectively and
is not a threshold selected from normal development data.

For the held-out scoring comparison, we also use
\begin{equation}
 \operatorname{nPAUC}_{\alpha}
 =\frac{1}{\alpha}\int_0^{\alpha}\operatorname{TPR}(u)\,du,
 \qquad \alpha=0.01.
 \label{eq:app_normalized_partial_auc}
\end{equation}
Here $u$ is the false-positive rate, $\operatorname{TPR}(u)$ is recall
on the interpolated ROC curve, and $\alpha$ is the upper false-positive
rate of interest. This quantity is the average recall over the first
1\% of the ROC curve. It is area divided by $\alpha$, not a standardized
partial AUC that assigns random ranking a value of $0.5$.

\paragraph{Aggregation and variability.}
Adapter summaries use optimization seeds 5, 1014, 2023, 3032, and 4041.
All five runs share the same pretrained base checkpoint. Where shown,
$\pm$ denotes the sample standard deviation across these five runs;
it does not measure uncertainty over independently pretrained detectors.
Reference-policy means give equal weight to each detector--dataset
pair. Numerical outcome bands used for reference comparisons are
descriptive categories, not statistical significance tests. Intervals for
reference effects come from a circular moving-block bootstrap of each
recording's scored timestamps with one-hour blocks and 200 resamples;
recordings are resampled independently when effects are combined across
pairs. These intervals describe within-recording sampling variability,
not variation across detectors, datasets, or retraining.

\section{Reference Interventions and Their Limits}
\label{app:reference}

A reference intervention asks whether the same residuals become more
informative when normal behavior is represented differently. Distribution
change is a common reason to reconsider a learned reference
\citep{gama2014drift}. The intervention must
account for two competing effects: an old reference can overstate
deviation during normal operation, while an adaptive reference can
incorporate an ongoing attack. The following results examine both effects.

\paragraph{Reference statistic and decision regions.}
The fixed reference consists of the training-normal residual mean and a
Ledoit--Wolf covariance estimate \citep{ledoitwolf2004}. For residual
dimension $p$, the nominal energy scale is $p/2$. This is the expected
CNRE when the reference equals the population normal mean and covariance;
an estimated shrinkage reference need not attain it exactly.
We compute $D_{\mathrm{ref}}$ as the median of the within-block median
energies, divided by $p/2$. Blocks span one hour in the primary
SWaT, WADI, and HAI study. Attack labels do not enter this calculation.
The normal-only curves in
Figure~\ref{fig:normal_cnre_shift} (Section~\ref{sec:exp_gdn_case})
are a separate, label-assisted visualization.

The empirical rule recommends retaining the reference when
$D_{\mathrm{ref}}<3$, investigating an update when
$D_{\mathrm{ref}}>10$, and makes no recommendation in the intervening
region. The three GDN cases informed these boundaries. The other twelve
pairs test transfer across detector families on the same three datasets.
No pair lies in the no-recommendation region, so the experiments do not
establish optimal boundaries or evaluate behavior within that region.
The diagnostic policy below retains the fixed reference when no update
is recommended. Table~\ref{tab:reference_diagnosis_main} lists the
diagnostic value, recommendation, and guarded-update effect for each
pair, and Figure~\ref{fig:reference_diagnostic_outcomes} in
Section~\ref{sec:exp_reference} plots them.

\begin{table}[t]
\centering
\small
\setlength{\tabcolsep}{5pt}
\caption{Reference diagnosis and guarded-update effects on all eligible
timestamps. $\Delta$AUC is guarded minus fixed AUC on the same eligible
timestamps. The GDN rows informed the decision boundaries. None of the
fifteen pairs lies in the 3--10 no-recommendation region. Values are
rounded. The 95\% interval comes from a circular moving-block bootstrap of
the scored timestamps (one-hour blocks, 200 resamples) and reflects
within-recording sampling variability only.}
\label{tab:reference_diagnosis_main}
\begin{tabular}{llrlrc}
\toprule
Dataset & Detector & $D_{\mathrm{ref}}$ & Recommendation & $\Delta$AUC & 95\% interval \\
\midrule
SWaT & GDN & $8.71\times10^3$ & Investigate update & $-0.012$ & [$-0.134$, $0.153$] \\
     & FuSAGNet & $72.6$ & Investigate update & $+0.073$ & [$0.027$, $0.146$] \\
     & TranAD & $2.15\times10^3$ & Investigate update & $-0.037$ & [$-0.128$, $0.059$] \\
     & NSIBF & $62.5$ & Investigate update & $-0.002$ & [$-0.100$, $0.106$] \\
     & GeCo & $0.181$ & Retain & $-0.369$ & [$-0.524$, $-0.152$] \\
\midrule
WADI & GDN & $7.77\times10^7$ & Investigate update & $+0.335$ & [$0.073$, $0.558$] \\
     & FuSAGNet & $2.75\times10^7$ & Investigate update & $+0.316$ & [$0.054$, $0.534$] \\
     & TranAD & $2.15\times10^7$ & Investigate update & $+0.204$ & [$-0.023$, $0.386$] \\
     & NSIBF & $1.94$ & Retain & $-0.063$ & [$-0.306$, $0.065$] \\
     & GeCo & $0.00486$ & Retain & $-0.097$ & [$-0.230$, $0.002$] \\
\midrule
HAI  & GDN & $0.895$ & Retain & $-0.054$ & [$-0.115$, $-0.020$] \\
     & FuSAGNet & $1.02$ & Retain & $-0.086$ & [$-0.166$, $-0.036$] \\
     & TranAD & $0.917$ & Retain & $-0.063$ & [$-0.124$, $-0.027$] \\
     & NSIBF & $0.850$ & Retain & $-0.013$ & [$-0.045$, $0.009$] \\
     & GeCo & $0.182$ & Retain & $-0.003$ & [$-0.010$, $0.002$] \\
\bottomrule
\end{tabular}
\end{table}

\paragraph{Causal rolling and guarded references.}
The rolling reference is refitted hourly using the preceding 12-hour
residual window, ending one hour before the current scoring block.
Within that fitting window, each residual coordinate is clipped at
its empirical 1st and 99th percentiles. The clipped residuals determine
the mean and Ledoit--Wolf covariance. This winsorization limits extreme
samples but does not use attack labels to identify contamination.
Fixed, rolling, and guarded scores are compared after the same initial
13-hour warm-up and on their common eligible timestamps.

The guarded version checks the preceding hour's alarm fraction. If more
than 20\% of that hour exceeds its trailing 99.5th-percentile threshold,
the next reference refit is skipped. The mean and covariance can remain
frozen for at most 24 consecutive hourly updates, after which refitting
is allowed. The operational alarm threshold is the training-normal
99th percentile for the fixed reference and the trailing-window 99th
percentile for an adaptive reference. The latter is recalculated under
the current mean and covariance even while those parameters are frozen.
Thus, this guard freezes the reference, not the entire alarm pipeline.

\paragraph{Policy comparison on all eligible timestamps.}
Table~\ref{tab:reference_policy_appendix} summarizes the consequences
of applying the diagnostic recommendation. The diagnostic policy uses
the guarded update only for update-diagnosed pairs. Its mean AUC is
$0.809$, compared with $0.751$ for a fixed reference and $0.759$ for
always applying the guarded update. Resampling each recording in one-hour
blocks (Appendix~\ref{app:protocol}), the policy exceeds never updating by
$0.058$ (95\% interval $0.016$--$0.109$) and always updating by $0.050$
($0.030$--$0.076$); on the twelve non-GDN pairs the margins are $0.046$
($0.014$--$0.086$) and $0.058$ ($0.033$--$0.090$). However, it still degrades two
pairs by at least $0.01$. The result supports selective intervention,
not a guarantee that every recommended update helps.

The policy uses a diagnostic calculated over the evaluation recording.
Although this calculation does not use attack labels, the policy
comparison is retrospective. It is not an online decision made before
the recording begins. Furthermore, the three GDN boundary-development
cases remain part of the fifteen-pair mean.

\begin{table}[t]
\centering
\small
\caption{Reference policies on all eligible timestamps. Means are
unweighted across detector--dataset pairs. Improvement and degradation
count changes of at least $+0.01$ and at most $-0.01$, respectively,
relative to never updating. The twelve non-GDN pairs share the same
three datasets as the boundary-development cases.}
\label{tab:reference_policy_appendix}
\begin{tabular}{lrrc}
\toprule
Policy & AUC: all 15 & AUC: non-GDN 12 & Improved / degraded \\
\midrule
Never update & 0.751 & 0.759 & 0 / 0 \\
Always guarded update & 0.759 & 0.747 & 4 / 9 \\
Diagnostic guarded update & 0.809 & 0.805 & 4 / 2 \\
\bottomrule
\end{tabular}
\end{table}

\paragraph{Why the guard is only a partial remedy.}
Table~\ref{tab:swat_reference_detail} retains all eligible SWaT timestamps.
FuSAGNet benefits from guarding: AUC rises from $0.709$ with ordinary
rolling to $0.920$ with guarding, above the fixed-reference value of
$0.848$. GDN recovers part of its rolling-reference loss but remains
below its fixed-reference AUC. TranAD does not improve with the guard.
GeCo, for which the diagnostic recommends retaining the reference,
loses substantial separation when its reference is updated.

The alarm columns of Table~\ref{tab:swat_reference_detail} separate
threshold adaptation from reference updating. For the four
update-diagnosed detectors, a fixed training-normal threshold flags about
$80\%$ of normal timestamps. A trailing 99th-percentile threshold on the
fixed reference lowers this to $2.0$--$10.4\%$; for NSIBF it keeps
$71.1\%$ recall, more than the guarded reference ($40.2\%$), whereas for
GDN and FuSAGNet the guarded reference reaches a lower FPR with higher
recall. A low adaptive FPR does not ensure high recall: the trailing
variants detect $1.0$--$71.2\%$ of attack timestamps.

\begin{table}[t]
\centering
\small
\caption{SWaT reference and threshold interventions on matched,
all-eligible support. ROC-AUC uses the fixed, rolling, or guarded
reference. Alarm columns give recall / normal FPR in percent for the
fixed reference with the training-normal 99th percentile, the fixed
reference with a trailing 12-hour 99th percentile, and the guarded
reference with its trailing 99th percentile.}
\label{tab:swat_reference_detail}
\begin{tabular}{lrrrccc}
\toprule
 & \multicolumn{3}{c}{ROC-AUC} & \multicolumn{3}{c}{Alarm recall / normal FPR (\%)} \\
\cmidrule(lr){2-4}\cmidrule(lr){5-7}
Detector & Fixed & Rolling & Guarded & Fixed, training & Fixed, trailing & Guarded \\
\midrule
GDN      & 0.814 & 0.771 & 0.802 & 93.3 / 79.6 & 16.4 / 10.4 & 35.1 / 1.3 \\
FuSAGNet & 0.848 & 0.709 & 0.920 & 94.6 / 79.2 & 19.5 / 2.0 & 71.2 / 1.4 \\
TranAD   & 0.428 & 0.391 & 0.391 & 93.4 / 78.3 & 1.0 / 2.2 & 7.1 / 2.1 \\
NSIBF    & 0.836 & 0.823 & 0.834 & 94.1 / 80.3 & 71.1 / 3.3 & 40.2 / 4.7 \\
GeCo     & 0.887 & 0.518 & 0.518 & 68.8 / 1.3 & 16.2 / 0.7 & 5.2 / 1.2 \\
\bottomrule
\end{tabular}
\end{table}

\paragraph{Sensitivity to long attacks and normalization.}
A supplementary evaluation removes the longest contiguous attack run
from each stream when calculating metrics. The references are still
constructed from the original stream: the excluded run is not removed
from rolling-reference fitting. Under this evaluation, mean AUC is
$0.700$ for the fixed reference, $0.779$ for always rolling, and
$0.808$ for diagnostic rolling. These numbers describe a different
evaluation support and a different update variant from
Table~\ref{tab:reference_policy_appendix}.

For the coarse outcome analysis, an AUC change above $+0.05$ is
classified as helpful, below $-0.05$ as harmful, and otherwise as neutral.
An update recommendation agrees with a helpful outcome; a retain
recommendation agrees with a neutral or harmful outcome.
Agreement is 15/15 after longest-run exclusion with ordinary rolling,
but 12/15 when long attacks are retained and the guard is used.
The three latter disagreements are GDN, TranAD, and NSIBF on SWaT:
their guarded effects are neutral rather than helpful. This outcome
definition is distinct from the $0.01$ counting rule in
Table~\ref{tab:reference_policy_appendix}.

A diagonal-normalization control retains each residual mean and
variance but removes cross-variable covariance. In the longest-run-
excluded analysis, it gives agreeing outcomes for all fourteen decided
pairs. NSIBF--SWaT changes from $D_{\mathrm{ref}}=62.5$ under full
covariance to $8.71$ under diagonal normalization and enters the
no-recommendation region. Thus, these reference results do not establish
that full covariance is essential for every diagnostic decision.

\subsection{Additional Dataset Checks}
\label{app:scope}

\paragraph{BATADAL.}
This check applies the existing decision boundaries to a separately
trained GDN on BATADAL \citep{taormina2018batadal}. The model is trained
on the organizers' first training dataset and evaluated on their second
training dataset, which contains seven attacks. This is not the
withheld official test set. The chronological normal split allocates
80\% to fitting and 20\% to validation. GDN uses a history of five
observations, hidden width 64, five graph neighbors, and dropout $0.2$.
Training uses Adam with learning rate $10^{-3}$, batch size 128, at
most 600 epochs, and early-stopping patience 20.

Because BATADAL observations are hourly, the diagnostic uses 24-hour
blocks. References are updated daily using a 168- or 336-hour trailing
window with a 24-hour lag. The decision boundaries remain 3 and 10.
All three training seeds receive a retain recommendation
(Table~\ref{tab:batadal_scope}). For the 168-hour comparison, the
matched-support mean AUC is $0.798$ for the fixed reference and $0.780$
for rolling. Seed 5 nevertheless improves by $0.037$. Agreement with
retain therefore refers to the coarse $0.05$ outcome band, not to the
claim that every update reduces AUC.

\begin{table}[t]
\centering
\small
\caption{BATADAL reference check on the organizers' second training
dataset. Each update effect is paired with the fixed reference on the
same timestamps. The two window lengths can have different eligible
supports because of their different warm-up periods.}
\label{tab:batadal_scope}
\begin{tabular}{rrlrr}
\toprule
Seed & $D_{\mathrm{ref}}$ & Recommendation & $\Delta$AUC: 168 h & $\Delta$AUC: 336 h \\
\midrule
5    & 2.45 & Retain & $+0.037$ & $+0.040$ \\
1014 & 1.06 & Retain & $-0.035$ & $-0.027$ \\
2023 & 0.89 & Retain & $-0.055$ & $-0.047$ \\
\bottomrule
\end{tabular}
\end{table}

\paragraph{NoBOOM.}
We also examine the ternary batch-distillation plant from NoBOOM
\citep{wagner2025noboom}. These recordings contain process faults and
should not be treated as interchangeable with the cyberattack datasets.
The pilot uses five-minute diagnostic blocks and a ten-minute rolling
window with a one-minute lag and one-minute refits. It clips fitting
residuals at the 1st and 99th percentiles.

For each operating point, GDN is fitted to its normal recordings using
one-second observations and a history of ten samples. Channels with
training standard deviation below $10^{-9}$ are removed; the remaining
channels use training-normal min--max scaling. The chronological split
uses 80\% for fitting and 20\% for validation. The model has hidden
width 64, up to five graph neighbors, and dropout $0.2$. Training uses
seed 5, Adam at $10^{-3}$, batches of 128, at most 199 epochs, and
early-stopping patience 10. The fixed reference uses fitting-split
residuals. Fault targets are identified by a positive value of the
dataset's common fault label.

After excluding each recording's longest fault run from evaluation,
24 recordings remain evaluable. Of the 17 that receive an update
recommendation, five improve, eight worsen, and four are neutral under
the $0.05$ outcome band; six fall in the no-recommendation region and
one is retained. This is a scope limitation, including under
the excluded-run protocol. The pilot changes both the recording
structure and the adaptation time scale. It therefore does not isolate
one cause of failure, and it does not justify transferring the
continuous-recording rule to short batches without further validation.

\section{Supplementary Attack-Evidence Results}
\label{app:events}

This appendix reports point-level CNRE evidence for the archived
detector outputs of the cross-detector study and a
supplementary interval-level view for the GDN case study.

\paragraph{Archived detector outputs.}
Table~\ref{tab:cross_detector_ap} reports point-level ROC-AUC and AP
for all fifteen archived detector--dataset pairs under each detector's
training-normal reference, on the full archived test support with all
attacks retained. For most SWaT and WADI pairs these values are
dominated by the reference mismatch of Section~\ref{sec:exp_reference};
the grouping in Table~\ref{tab:attribution} reads them together with the diagnosed
reference. Table~\ref{tab:hai_archived_evidence} adds $R_{1\%}$ for
HAI, where the reference is retained: TranAD has the largest ROC-AUC
and $R_{1\%}$, while FuSAGNet has the largest AP. These values describe
archived residual artifacts and references, not a retraining comparison
under identical preprocessing. They need not match the GDN case
study, whose fitting split and eligible targets differ.

\begin{table}[t]
\centering
\small
\setlength{\tabcolsep}{5pt}
\caption{Point-level CNRE evidence under fixed training-normal references
for all fifteen pairs. Each pair uses its full archived test support,
retaining all attacks. Preprocessing, sampling intervals, and attack
prevalence can differ between detector outputs; this is a diagnostic
summary of the artifacts. The AP of a random score equals the attack
prevalence: 0.121--0.124 on SWaT, 0.058--0.074 on WADI, and
0.022--0.036 on HAI. Unlike Table~\ref{tab:attribution}, this table
includes timestamps before the 13-hour warm-up of the adaptive references.}
\label{tab:cross_detector_ap}
\begin{tabular}{lrrrrrr}
\toprule
 & \multicolumn{2}{c}{SWaT} & \multicolumn{2}{c}{WADI}
 & \multicolumn{2}{c}{HAI} \\
\cmidrule(lr){2-3}\cmidrule(lr){4-5}\cmidrule(lr){6-7}
Detector & AUC & AP & AUC & AP & AUC & AP \\
\midrule
GDN      & 0.7877 & 0.6320 & 0.5302 & 0.1998 & 0.8897 & 0.5685 \\
FuSAGNet & 0.8310 & 0.7192 & 0.5295 & 0.2056 & 0.8985 & 0.6512 \\
TranAD   & 0.4754 & 0.1071 & 0.6117 & 0.3427 & 0.9095 & 0.5850 \\
NSIBF    & 0.8083 & 0.7000 & 0.7958 & 0.4419 & 0.8441 & 0.5595 \\
GeCo     & 0.8770 & 0.7540 & 0.6236 & 0.0779 & 0.8679 & 0.5535 \\
\bottomrule
\end{tabular}
\end{table}

\begin{table}[t]
\centering
\small
\caption{Point-level CNRE evidence for the archived HAI detector outputs,
each with its training-normal reference.}
\label{tab:hai_archived_evidence}
\begin{tabular}{lrrr}
\toprule
Detector & ROC-AUC & AP & $R_{1\%}$ \\
\midrule
GDN & 0.8897 & 0.5685 & 0.6578 \\
FuSAGNet & 0.8985 & 0.6512 & 0.7134 \\
TranAD & 0.9095 & 0.5850 & 0.7227 \\
NSIBF & 0.8441 & 0.5595 & 0.5246 \\
GeCo & 0.8679 & 0.5535 & 0.5521 \\
\bottomrule
\end{tabular}
\end{table}

\paragraph{Interval-level view.}
Point-level metrics are the primary evidence measurements in this
paper. An interval-level view supplements them by showing whether
pooled separation is shared across attacks or concentrated in a subset.
It uses the same scalar CNRE and fixed normal reference as the
point-level analysis.

\paragraph{Definition and interval construction.}
Let $A_k$ be the eligible timestamps in the $k$th contiguous labeled
attack interval, and let $N$ contain all eligible normal targets in
the same GDN case-study evaluation. We define
\begin{equation}
 S(A_k)=\frac{1}{|A_k||N|}
 \sum_{t\in A_k}\sum_{u\in N}
 \left(\mathbb{I}[e_t>e_u]+\frac{1}{2}\mathbb{I}[e_t=e_u]\right).
 \label{eq:app_scalar_event_auc}
\end{equation}
Here $e_t$ and $e_u$ are attack and normal CNRE values under the same
reference. The statistic is the fraction of attack--normal pairs in
which attack energy is larger, with half credit for ties. It is the
ROC-AUC obtained by comparing one attack interval against the common
normal pool. A large value indicates that this interval usually has
higher energy than normal targets in that pool.

For frozen GDN on HAI, we split labeled attack runs at recording
boundaries and gaps in eligible sample identifiers. All intervals are
retained, including short intervals. This gives 50 intervals containing
935 attack targets. Every interval uses the same 39,236 normal targets
as the primary point-level evaluation of the GDN case study. No interval-specific
reference or threshold is fitted.

\begin{figure}[t]
\centering
\includegraphics[width=\linewidth]{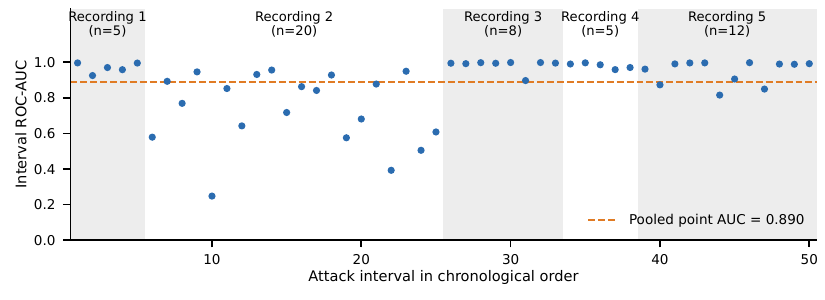}
\caption{Descriptive interval AUCs for frozen GDN on HAI.
Every interval is compared with the same full pool of eligible normal
targets. All intervals, including short intervals, are retained and
split at sample gaps and recording boundaries. The dashed line shows
pooled point-level AUC, which is the duration-weighted average of these
interval AUCs. No duration-matched null test or interval-level significance
label is used.}
\label{fig:scalar_event_evidence}
\end{figure}

\paragraph{Heterogeneity within the point-level result.}
Table~\ref{tab:hai_scalar_events} and
Figure~\ref{fig:scalar_event_evidence} show substantial variation.
The interval AUCs range from $0.2471$ to $0.9987$, while the pooled
point-level AUC is $0.8898$. Recording 2 contains the lowest-separation
interval and has a median interval AUC of $0.8053$. Most intervals
in the other recordings rank well against the common normal pool.
Thus, strong pooled evidence does not imply equally strong evidence
for every labeled attack.

\begin{table}[t]
\centering
\small
\caption{Descriptive interval AUC for frozen GDN on HAI.
Minima, medians, and maxima summarize intervals without duration
weighting. All rows use the same full normal comparison pool.}
\label{tab:hai_scalar_events}
\begin{tabular}{lrrrrr}
\toprule
Recording & Intervals & Attack targets & Minimum & Median & Maximum \\
\midrule
1 & 5  & 66  & 0.9255 & 0.9706 & 0.9966 \\
2 & 20 & 362 & 0.2471 & 0.8053 & 0.9564 \\
3 & 8  & 159 & 0.8973 & 0.9952 & 0.9987 \\
4 & 5  & 120 & 0.9589 & 0.9861 & 0.9969 \\
5 & 12 & 228 & 0.8155 & 0.9895 & 0.9980 \\
\midrule
All & 50 & 935 & 0.2471 & 0.9531 & 0.9987 \\
\bottomrule
\end{tabular}
\end{table}

The relation to the point-level result is exact when the intervals
partition the attack targets and share the same normal pool:
\begin{equation}
 S_{\mathrm{atk}}=
 \sum_{k=1}^{K}\frac{|A_k|}{\sum_{j=1}^{K}|A_j|}\,S(A_k).
 \label{eq:app_point_event_relation}
\end{equation}
Here $K$ is the number of intervals and $S_{\mathrm{atk}}$ is the pooled
point-level ROC-AUC of Eq.~(\ref{eq:point_attack_evidence}). Each interval contributes in proportion to its
number of eligible attack targets. At a fixed sampling interval this
is duration weighting. It explains why pooled AUC differs from the
unweighted median interval AUC, which is $0.9531$ in this evaluation.

\paragraph{Interpretation and limits.}
These interval scores are descriptive. They do not compare an attack
against a duration-matched distribution of normal windows, and we do
not use them to assign statistical evidence-present or evidence-absent
labels. In particular, a short interval can have an extreme AUC without
providing a stable estimate of attack behavior. An AUC below $0.5$
means that CNRE ranks this interval below the normal pool;
it does not establish that the residual vector contains no information.

Likewise, a high interval AUC does not ensure an alarm at a threshold
calibrated on normal development data. The primary extraction analysis
therefore remains the point-level scoring and calibration comparison
in Appendix~\ref{app:extraction}. The interval view identifies
heterogeneous cases for inspection without changing that diagnostic
criterion or making a claim about temporal-accumulation benefits.

\section{Scoring Rules and Alarm Calibration}
\label{app:extraction}

Evidence extraction concerns what happens after the residuals have been
produced. CNRE provides a direct scalar score, $s_t=e_t$, against which
we compare other scoring rules applied to the same residuals. A scoring
rule produces an anomaly score; the alarm decision rule compares that score
with a threshold. The comparisons below separate score ranking from the
behavior of a threshold fixed on normal data. A score can rank attacks well while a transferred threshold
produces an unsuitable false-positive rate.

\paragraph{Cross-detector trailing-percentile control.}
\label{app:quantile_control}
This control retains the pretrained detector, residuals, and
training-normal CNRE mean and covariance.
For each scoring block $b$, let $\mathcal{H}_b$ contain the
timestamps in the preceding 12-hour window, ending one hour
before the block begins. For each timestamp $t$ in that block,
the rescored value is
\begin{equation}
 s_t^{\mathrm{local}}
 =
 \frac{1}{|\mathcal{H}_b|}
 \sum_{u\in\mathcal{H}_b}
 \mathbb{I}[e_u\leq e_t].
 \label{eq:app_local_percentile}
\end{equation}
Here $e_t$ and $e_u$ are energies computed under the same fixed
CNRE reference, and $|\mathcal{H}_b|$ is the number of historical
samples. The score is the fraction of historical energies no
larger than the current energy. Including ties matches the
saved implementation's empirical cumulative distribution.
The history window advances hourly and uses all past residuals,
without filtering attack timestamps.

Fixed CNRE and the rescored values are evaluated on the same
timestamps after the initial 13-hour warm-up.
Table~\ref{tab:cross_detector_recalibration} reports two evaluation
supports: all scored timestamps, and the same timestamps without each
recording's longest contiguous attack run. In the second, the excluded
run remains in the original stream and can enter later history windows;
labels determine this evaluation mask but do not construct the scores.
The longest SWaT run lasts about ten hours and covers about $70\%$ of
the scored attack timestamps, so a 12-hour history window absorbs much
of it, which is consistent with the losses in the All columns. Without
the longest run, unguarded rolling updates change AUC by $+0.328$,
$+0.209$, $+0.125$, and $+0.220$ for GDN, FuSAGNet, TranAD, and NSIBF on
SWaT, and by $+0.283$, $+0.274$, and $+0.184$ for GDN, FuSAGNet, and
TranAD on WADI. These are the denominators of the recovered shares
in Section~\ref{sec:exp_extraction}.

Because the empirical distribution changes between blocks,
this transformation can change ranking across time.
It should not be described as selecting a different constant
threshold for an otherwise unchanged score.
The comparison evaluates scalar rescoring; it does
not compare all detectors' native scoring pipelines.

\paragraph{GDN-style scoring rule on fixed residuals.}
The GDN-style comparator follows the channel-normalization and
aggregation structure of GDN \citep{deng2021gdn}, with its normalization
statistics fitted on the normal fitting split. For residual coordinate
$r_{t,j}$, define
\begin{equation}
 u_{t,j}=\frac{|r_{t,j}|-m_j}{q_{0.75,j}-q_{0.25,j}+0.01},
 \qquad
 s_t^{\mathrm{GDN}}=
 \max_{j=1,\ldots,p}\left(\frac{1}{4}\sum_{k=0}^{3}u_{t-k,j}\right).
 \label{eq:app_gdn_scoring_rule}
\end{equation}
Here $j$ indexes the $p$ residual coordinates. The quantities $m_j$,
$q_{0.25,j}$, and $q_{0.75,j}$ are the median, first quartile, and third
quartile of $|r_{t,j}|$ on normal fitting data. The constant $0.01$
stabilizes normalization when a channel has a small interquartile range.
The four-sample trailing average smooths each normalized channel, and
the maximum selects the largest smoothed channel response.

This equation describes the comparator used here rather than every
native GDN evaluation procedure. In that implementation, the first three
scores are zero. Smoothing follows the concatenated eligible sample
sequence and does not reset at recording boundaries or removed-window gaps.
We preserve this behavior when reproducing the archived comparison.
CNRE and the GDN-style scoring rule are evaluated on identical saved
residuals and target labels (Table~\ref{tab:extraction_main}). Their difference concerns the scoring rule as
a whole; it does not isolate covariance normalization from smoothing
or channel aggregation.

\begin{figure}[t]
\centering
\includegraphics[width=\linewidth]{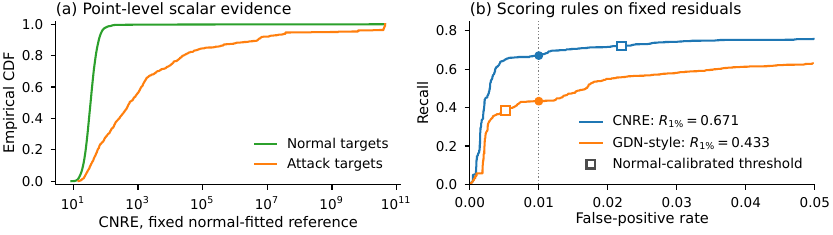}
\caption{Point-level evidence and extraction for the frozen GDN on HAI in
the GDN case study. (a) Empirical CNRE distributions for eligible
attack and normal targets under the same normal-fitting reference.
(b) Low-FPR ROC curves for CNRE and GDN-style scores on identical
residuals and timestamps. Filled circles show retrospective recall at
1\% test FPR, read from the test ROC curve; open squares show where each
score's threshold calibrated on normal development data lands
(Table~\ref{tab:hai_fixed_thresholds}).}
\label{fig:hai_evidence_extraction}
\end{figure}

\paragraph{Thresholds fixed on normal development data.}
For each model and score, the fixed threshold is the 99.5th percentile
on the early-stopping normal split. The quantile uses the upper empirical
order statistic (the \texttt{higher} convention), and an alarm requires
a score strictly above that threshold. No test labels select this
threshold. Table~\ref{tab:hai_fixed_thresholds} reports the realized
recall and false-positive rate on all HAI test timestamps of the GDN case study.

\begin{table}[t]
\centering
\small
\setlength{\tabcolsep}{5pt}
\caption{Score ranking on the same frozen GDN HAI residuals. $R_{1\%}$
uses test ROC interpolation. Table~\ref{tab:hai_fixed_thresholds} gives
the alarms with thresholds fixed on normal development data.}
\label{tab:extraction_main}
\begin{tabular}{lrrr}
\toprule
Scoring rule & ROC-AUC & AP & $R_{1\%}$ \\
\midrule
GDN-style & 0.8425 & 0.3768 & 0.4332 \\
CNRE & 0.8898 & 0.5734 & 0.6706 \\
\bottomrule
\end{tabular}
\end{table}

\begin{table}[t]
\centering
\small
\caption{HAI alarms with thresholds fixed on normal development data.
FPR counts all normal-labeled eligible test targets. Adapter rows
average five optimization seeds. These are different operating points,
not comparisons at a matched test false-positive rate.}
\label{tab:hai_fixed_thresholds}
\begin{tabular}{llrr}
\toprule
Model & Score & Recall & Test FPR \\
\midrule
Frozen GDN & CNRE & 0.721 & 0.0219 \\
Primary adapter & CNRE & 0.716 & 0.0375 \\
Frozen GDN & GDN-style & 0.384 & 0.0051 \\
Primary adapter & GDN-style & 0.422 & 0.0399 \\
\bottomrule
\end{tabular}
\end{table}

For frozen GDN, CNRE detects more attack targets but also flags more
normal targets at these separately calibrated thresholds. Section~\ref{sec:exp_gdn_case} and
Figure~\ref{fig:hai_evidence_extraction} therefore also compare
their retrospective ROC curves and $R_{1\%}$ on the same residuals. Those matched-FPR comparisons address
ranking quality; Table~\ref{tab:hai_fixed_thresholds} addresses
normal-to-test threshold transfer. Neither measurement replaces the other.

\paragraph{Held-out scoring-rule selection.}
A separate experiment asks whether a more elaborate scoring rule improves
on CNRE when selected on earlier recordings. The scorer-selection
recordings, HAI 1--2, provide 16,189 eligible targets, and recordings
3--5 provide 23,982 held-out targets. The residual-producing GDN and its normal-fitting
CNRE reference stay fixed. Scorer statistics are estimated only from
clean normal samples of the scorer-selection recordings. Here, clean excludes
normal targets whose model-input histories contain an attack.

The candidate with the largest point-level selection AUC is frozen
before held-out evaluation. Since selection AUC uses attack labels,
this is a label-assisted benchmark intervention. It is not a
normal-only deployment procedure. All target labels, including normal
targets with attack-containing histories, remain in the point-level
AUC calculations.

\paragraph{Candidate definitions.}
Let $\Sigma_0=V\Lambda V^\top$ be the eigendecomposition of the fixed
normal covariance, with eigenvalues in descending order. We compute
\begin{equation}
 z_t=\Lambda^{-1/2}V^\top(r_t-\mu_0),
 \qquad E_{t,j}=z_{t,j}^2,
 \qquad \sum_{j=1}^{p}E_{t,j}=2e_t.
 \label{eq:app_scoring_coordinates}
\end{equation}
The columns of $V$ are covariance eigenvectors, $\Lambda$ contains
their eigenvalues, and $\mu_0$ is the normal residual mean. Each
$z_{t,j}$ is a covariance-normalized residual coordinate, while
$E_{t,j}$ is its squared contribution to total energy. The sum score
has exactly the same ranking as CNRE. Coordinates are used here to
specify alternative downstream scoring rules, not to introduce a separate
directional attack-evidence diagnosis.

The ten candidates are the sum, maximum, top-three sum, and top-five
sum of $E_{t,j}$; robust sum and robust maximum; two weighted sums;
and sum and maximum over the lower-eigenvalue half of the coordinates.
The top-$k$ rules sum the $k$ largest values at each timestamp.
The lower-eigenvalue rules use
$j=\lfloor p/2\rfloor+1,\ldots,p$; for HAI this retains 40 of 79
coordinates.

For the robust rules, let $b_j$ be the median of $z_{t,j}$ on the
scorer-selection clean-normal pool and $d_j$ its median absolute deviation
about $b_j$. Then
\begin{equation}
 \widetilde z_{t,j}=\frac{z_{t,j}-b_j}{1.4826d_j+10^{-9}},
 \qquad
 s_t^{\mathrm{robust}}=\sum_{j=1}^{p}\widetilde z_{t,j}^{2}.
 \label{eq:app_robust_scoring_rule}
\end{equation}
The normal-consistency factor $1.4826$ rescales the median absolute
deviation, and $10^{-9}$ prevents division by zero. The robust-maximum
candidate uses $\max_j|\widetilde z_{t,j}|$ instead of the sum of
squares. These rules fit additional centering and scaling parameters on the
scorer-selection recordings while leaving the original residual producer
and the CNRE whitening transform unchanged.

For the instability-weighted sum, the clean selection indices are
split chronologically into twenty approximately equal groups. Let
$m_{g,j}$ and $v_{g,j}$ denote the mean and population variance of
$z_{t,j}$ in group $g$. The unnormalized weight is
\begin{equation}
 \widetilde w_j=\left[
 \operatorname{Var}_{g}(m_{g,j})+
 \frac{1}{20}\sum_{g=1}^{20}|v_{g,j}-1|+10^{-3}
 \right]^{-1}.
 \label{eq:app_instability_weights}
\end{equation}
The two terms penalize changes in the normal coordinate mean and
departures from unit variance. The constant $10^{-3}$ stabilizes the
inverse. We divide these weights by their mean across coordinates
and use $\sum_j w_jE_{t,j}$ as the score.
The kurtosis-weighted candidate instead starts from
$\widetilde w_j=1/\max(\gamma_j,0.5)$, where $\gamma_j$ is the
selection-normal excess kurtosis of coordinate $j$. It is computed
from the fourth central moment divided by the squared population
variance plus $10^{-12}$, minus three. Its weights are also normalized
to have mean one.

\begin{table}[t]
\centering
\small
\caption{HAI scoring candidates. Robust sum was selected by AUC on
recordings 1--2. Other held-out results were examined only after this choice
was frozen and do not define a new selected rule.}
\label{tab:hai_scoring_candidates}
\begin{tabular}{lrr}
\toprule
Rule & Selection AUC & Held-out AUC \\
\midrule
Robust sum (selected) & 0.8456 & 0.9539 \\
Sum (CNRE ranking) & 0.8430 & 0.9572 \\
Top-three sum & 0.8424 & 0.9494 \\
Maximum & 0.8423 & 0.9419 \\
Top-five sum & 0.8423 & 0.9522 \\
Kurtosis-weighted sum & 0.8376 & 0.9393 \\
Robust maximum & 0.8362 & 0.9396 \\
Lower-eigenvalue half: maximum & 0.8350 & 0.9504 \\
Lower-eigenvalue half: sum & 0.8335 & 0.9599 \\
Instability-weighted sum & 0.7994 & 0.8751 \\
\bottomrule
\end{tabular}
\end{table}

\paragraph{Held-out ranking and calibration.}
The selected robust sum improves selection AUC from $0.8430$ to
$0.8456$, but its held-out AUC is $0.9539$, below CNRE's $0.9572$
(Table~\ref{tab:hai_scoring_candidates}). Held-out
$\operatorname{nPAUC}_{0.01}$ also falls, from $0.6410$ to $0.6292$.
The lower-eigenvalue-half sum has a higher held-out AUC, but choosing it
after viewing that result would violate the selection protocol.

For alarms in this experiment, each score uses its empirical 99th
percentile on scorer-selection clean normals, with NumPy's linear
quantile interpolation and strict threshold exceedance. The held-out
clean-normal FPR is $0.0027$ for CNRE and $0.0068$ for robust sum.
These rates exclude attack-containing input histories and therefore
use a narrower normal denominator than
Table~\ref{tab:hai_fixed_thresholds}. They should not be silently
substituted for the all-normal-target FPR in the primary metrics.

The positive fixed-residual comparison and the unsuccessful held-out
selection answer different questions. The first shows that the scoring rule
can limit extraction of available residual evidence. The second shows
that changing the scoring rule does not guarantee improvement on later
recordings. These results motivate evaluating scoring interventions on
held-out timestamps and reporting their realized false-positive rates.

\section{Adapter Case Study and Reference Sensitivity}
\label{app:adapter}

The adapter is a controlled model intervention. Its purpose is
to test whether improving normal prediction also improves attack
evidence. It is not trained to optimize attack detection, and its
results do not establish a limit on all possible representation changes.

\paragraph{Architecture and insertion site.}
The primary adapter is inserted after GDN's node-embedding multiplication,
at a late representation with 64 features per node. It applies the same
residual transformation to every node:
\begin{equation}
 h' = h + W_{\mathrm{up}}
       \operatorname{GELU}\!\left(
          W_{\mathrm{down}}\operatorname{LN}(h)+b_{\mathrm{down}}
       \right)+b_{\mathrm{up}}.
 \label{eq:app_adapter}
\end{equation}
Here $h\in\mathbb{R}^{64}$ is a node's original hidden representation
and $h'$ is its adapted representation. Layer normalization
$\operatorname{LN}$ uses no trainable affine parameters and numerical
stabilizer $10^{-5}$. The matrix
$W_{\mathrm{down}}\in\mathbb{R}^{w\times64}$ compresses the normalized
features to bottleneck width $w$, and
$W_{\mathrm{up}}\in\mathbb{R}^{64\times w}$ maps the nonlinear correction
back to the original dimension. The vectors $b_{\mathrm{down}}$ and
$b_{\mathrm{up}}$ are the corresponding biases. GELU supplies the
nonlinearity. Adding the correction to $h$ lets the model begin from
the original detector and learn a restricted change to its prediction.

The primary width is $w=16$, giving 2,128 trainable parameters.
The upward projection and its bias are initialized to zero, so the
initial adapter is the identity. The base detector's parameters and
buffers remain frozen. Its modules stay in evaluation mode, so base
dropout is disabled and normalization buffers do not change. The
adapter itself has no dropout.

\paragraph{Training and reference fitting.}
Only normal fitting windows provide gradient updates. The objective
is prediction MSE; there is no attack supervision or residual-energy
regularizer. Adam \citep{kingma2015adam} uses learning rate $10^{-3}$,
batch size 256, coefficients $(0.9,0.999)$, numerical stabilizer
$10^{-8}$, and zero weight decay. Runs last at most 120 epochs, with
early-stopping patience 15. The identity model at epoch zero remains
eligible for selection. The five optimization seeds are listed in
Appendix~\ref{app:protocol}.

After checkpoint selection, each model receives its own CNRE mean and
covariance fitted to that model's residuals on the normal fitting split.
This keeps the data used for reference fitting consistent while allowing
the residual distribution to change after adaptation. The experiment
therefore evaluates the adapted residuals together with their own normal
reference; it does not hold the numerical covariance matrix fixed across
different residual-producing models. Table~\ref{tab:hai_adapter_reference}
reports that alternative for HAI.

\paragraph{HAI controls.}
For the primary adapter, mean selection-normal MSE falls from
$5.19\times10^{-3}$ to $3.77\times10^{-3}$, a reduction of approximately 27.3\%.
However, CNRE AP falls from $0.5734$ to $0.4523$, and retrospective
$R_{1\%}$ falls from $0.6706$ to $0.5037$.
Figure~\ref{fig:adapter_normal_fit_evidence} shows the five
individual runs under both reference protocols.
Table~\ref{tab:hai_adapter_reference} contrasts this protocol with
retaining the frozen model's reference, as discussed in
Section~\ref{sec:exp_gdn_case}. The small change in overall AUC does not capture the larger
loss at low false-positive rates.

Table~\ref{tab:hai_adapter_controls} provides nearby controls: smaller
bottlenecks at the same insertion site, a width-16 adapter after the
output normalization and activation, and training the existing output
head alone. These alternatives do not reverse the primary conclusion
for AP or $R_{1\%}$. Their role is to check the scope of the case study,
not to support a claim that all adapter architectures fail.

\begin{table}[t]
\centering
\small
\setlength{\tabcolsep}{4pt}
\caption{HAI point-level CNRE results on identical test targets of the
GDN case study.
Each model uses its own normal-fitting reference. Adapted rows report
mean $\pm$ sample standard deviation over five optimization seeds.
The frozen row is one shared checkpoint.}
\label{tab:hai_adapter_controls}
\begin{tabular}{lccc}
\toprule
Intervention & ROC-AUC & AP & $R_{1\%}$ \\
\midrule
Frozen GDN & $0.8898$ & $0.5734$ & $0.6706$ \\
Output head only & $0.8851\pm0.0003$ & $0.5635\pm0.0004$ & $0.6539\pm0.0006$ \\
Late adapter, $w=4$ & $0.8847\pm0.0046$ & $0.4848\pm0.0140$ & $0.5621\pm0.0192$ \\
Late adapter, $w=8$ & $0.8856\pm0.0035$ & $0.4712\pm0.0155$ & $0.5330\pm0.0257$ \\
Late adapter, $w=16$ & $0.8877\pm0.0028$ & $0.4523\pm0.0120$ & $0.5037\pm0.0284$ \\
Post-activation, $w=16$ & $0.8846\pm0.0031$ & $0.4722\pm0.0139$ & $0.5521\pm0.0232$ \\
\bottomrule
\end{tabular}
\end{table}

\begin{figure}[t]
\centering
\includegraphics[width=\linewidth]{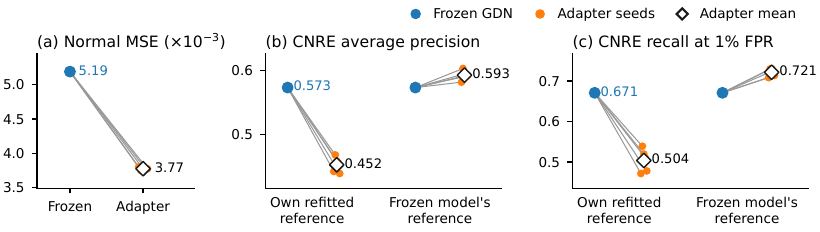}
\caption{Normal fit and point-level attack evidence for the primary HAI
adapter under two reference protocols. Blue markers denote the frozen GDN
checkpoint, orange markers the five adapter optimization seeds sharing
that checkpoint, and white diamonds their means. (a) Normal MSE on the
configuration-selection split. (b, c) CNRE AP and recall at 1\% FPR on
the test targets of the GDN case study, when each model uses a reference
refitted to its own normal residuals and when both use the frozen
model's reference. Vertical axes do not start at zero.}
\label{fig:adapter_normal_fit_evidence}
\end{figure}

\paragraph{HAI reference protocols.}
The primary comparison refits each model's CNRE mean and covariance to
that model's residuals on the normal fitting split.
Table~\ref{tab:hai_adapter_reference} scores the same models under
three other references. The frozen model's reference applies the mean
and covariance fitted to the frozen GDN residuals to both models, so it
holds the numerical reference fixed and uses only normal training data.
The test-regime reference is fitted to clean test-normal residuals
before the median clean-normal index and scores only the later half; the
oracle reference uses all clean test-normal residuals. Both use test
labels retrospectively. Under the frozen model's reference,
the adapter improves all three metrics, and every seed has a higher AP
than the frozen model. Under both test-normal references, its AP and
$R_{1\%}$ losses are much smaller than under the primary protocol.
The sign and size of the adapter effect therefore depend on the
reference-fitting protocol; normal-fit improvement alone does not
determine them.

\begin{table}[t]
\centering
\small
\setlength{\tabcolsep}{4pt}
\caption{HAI adapter results under four reference protocols. Adapter
entries are five-seed means. Frozen and adapted models share evaluation
timestamps within each row; the test-regime row scores only the later
half of the test stream.}
\label{tab:hai_adapter_reference}
\begin{tabular}{lrrrrrr}
\toprule
 & \multicolumn{2}{c}{ROC-AUC} & \multicolumn{2}{c}{AP} & \multicolumn{2}{c}{$R_{1\%}$} \\
\cmidrule(lr){2-3}\cmidrule(lr){4-5}\cmidrule(lr){6-7}
Reference & Frozen & Adapter & Frozen & Adapter & Frozen & Adapter \\
\midrule
Own refit (primary) & 0.8898 & 0.8877 & 0.5734 & 0.4523 & 0.6706 & 0.5037 \\
Frozen model's & 0.8898 & 0.9002 & 0.5734 & 0.5928 & 0.6706 & 0.7213 \\
Test-regime, later half & 0.9159 & 0.9153 & 0.6092 & 0.5820 & 0.6529 & 0.6059 \\
Oracle test-normal & 0.8547 & 0.8624 & 0.5079 & 0.4963 & 0.5797 & 0.5549 \\
\bottomrule
\end{tabular}
\end{table}

\paragraph{Per-recording normal energy.}
The HAI test stream consists of five recordings.
Table~\ref{tab:hai_adapter_recording_energy} gives the 99th percentile of
CNRE over each recording's normal timestamps. Under each adapter's own
refitted reference it rises about ninefold to twelvefold in recording 1
and fourfold to sixfold in recording 4 relative to the frozen GDN, but by
less than $1.7$ times in the other three recordings. Under the frozen
model's reference, all five recordings stay within $274$--$430$. This
uneven inflation is consistent with the lower precision under the primary
protocol; it does not separate recentering from covariance reweighting.

\begin{table}[t]
\centering
\small
\caption{99th percentile of CNRE over the normal timestamps of each HAI
test recording. Adapter columns use each adapter's own refitted reference
or the frozen model's reference and give the range over the five seeds.}
\label{tab:hai_adapter_recording_energy}
\begin{tabular}{rrrrr}
\toprule
Recording & Normal timestamps & Frozen GDN & Adapter, own refit & Adapter, frozen reference \\
\midrule
1 & 4,249  & 171 & 1511--2017 & 310--418 \\
2 & 11,512 & 81  & 99--108    & 274--398 \\
3 & 10,635 & 110 & 135--170   & 302--404 \\
4 & 3,834  & 107 & 456--660   & 315--430 \\
5 & 9,006  & 131 & 169--215   & 329--423 \\
\bottomrule
\end{tabular}
\end{table}

\paragraph{SWaT: a reference-sensitive companion.}
The same primary adapter reduces mean normal MSE from $1.71\times10^{-3}$
to $1.45\times10^{-3}$, approximately 15.6\%. With each model's own normal-fit
reference, whole-stream AUC falls from approximately $0.764$ to $0.664$.
We then change the reference protocol to examine whether this large
gap persists when the reference represents test-normal behavior.

The test-regime and oracle references are defined as for HAI. Both
require test labels and are retrospective diagnostic comparisons, not
deployable normal-training-only interventions.

\begin{table}[t]
\centering
\small
\caption{Reference sensitivity of the SWaT adapter case study.
Adapted results are five-seed mean AUCs. Frozen and adapted models
share evaluation timestamps within each row. Different rows need not
share the same support or data access.}
\label{tab:swat_adapter_reference}
\begin{tabular}{p{0.62\linewidth}rr}
\toprule
Reference and evaluation support & Frozen & Adapter \\
\midrule
Own normal-fit reference; full test & 0.764 & 0.664 \\
Test-regime reference; later-half test & 0.912 & 0.904 \\
Oracle test-normal reference; full test & 0.907 & 0.904 \\
\bottomrule
\end{tabular}
\end{table}

The much smaller within-row gaps under test-normal references show
that the large own-reference gap is sensitive to reference choice.
They do not identify a universal cause of adapter degradation.
In particular, subtracting values across rows would mix different
evaluation supports and reference-fitting data. Together with HAI,
this case study suggests checking reference validity before interpreting
a model intervention, and evaluating attack separation directly
rather than using improved normal MSE as its proxy.

\end{document}